\documentclass{amsart}

\usepackage{setspace} 

\usepackage{amsmath}
\usepackage{graphicx}
\usepackage{natbib}
\usepackage{url} 

\usepackage{multirow}%
\usepackage{amssymb,amsfonts,amsthm,mathrsfs,mathtools,dsfont}
\usepackage{xcolor}%
\usepackage{textcomp}%
\usepackage{tabularx}%
\usepackage{booktabs}%
\usepackage{algpseudocode}%
\usepackage{algorithm}
\usepackage{listings}%
\usepackage{stmaryrd}
\usepackage{tikz}
\graphicspath{figures}

\setcitestyle{authoryear,open={(},close={)}}
\usepackage{subcaption}
\usepackage{comment}

\newtheorem{theo}{Theorem}[section]

\newtheorem{prop}[theo]{Proposition}
\newtheorem{propri}[theo]{Property}
\newtheorem{remark}[theo]{Remark}
\newtheorem{ass}[theo]{Assumption}

\newcommand{\argmin}[1]{\underset{#1}{\mathrm{argmin} \ }}%
\renewcommand{\d}{\text{d}}%
\newcommand{\one}{\mathds{1}}%
\newcommand{\cM}{\mathcal{M}}%
\newcommand{\E}{\mathbb{E}}

\newcommand{\cO}{\mathcal{O}}%
\newcommand{\Gam}{\mathcal{G}\text{amma}}%
\newcommand{\R}{\mathbb{R}}%
\newcommand{\w}{\widehat}%
\newcommand{\taubf}{\text{\mathversion{bold}{$\tau$}}}%
\newcommand{\tbfp}{{\bf{tp}}}%
\newcommand{\abf}{{\bf{a}}}%
\newcommand{\cbf}{{\bf{c}}}%
\newcommand{\lambdabf}{\text{\mathversion{bold}{$\lambda$}}}%
\newcommand{\lambdabar}{\overline{\lambda}}%
\newcommand{\nubf}{\text{\mathversion{bold}{$\nu$}}}%
\newcommand{\rhobf}{\text{\mathversion{bold}{$\rho$}}}%
\newcommand{\minim}[1]{\underset{#1}{\mathrm{min} \ }}%

\definecolor{darkgreen}{RGB}{0,100,0}

\newcommand{\change}[2]{\textcolor{black}{#2}}
\newcommand{\changeNew}[2]{\textcolor{gray}{}\textcolor{black}{#2}}
\newcommand{\changeAgain}[2]{\textcolor{gray}{}\textcolor{black}{#2}}

\begin{document}


  \title{Multiple change-point detection for \changeNew{some point}{Poisson} processes}

 \author[label1]{Charlotte Dion-Blanc}
 \author[label1,label3]{Diala Hawat}
 \author[label2]{Emilie Lebarbier}
 \author[label1]{St\'{e}phane Robin}
 
%
%

\address[label1]{Sorbonne Universit\'{e}, Universit\'{e} Paris Cit\'{e}, CNRS, Laboratoire de Probabilit\'{e}s, Statistique et Mod\'{e}lisation, LPSM, 75005 Paris, France}

\address[label2]{Universit\'{e} Paris Nanterre, Modal'X, 92000 Nanterre, France}

\address[label3]{EDF, 91120 Saclay, France}
 
 \maketitle

\begin{abstract}
The aim of change-point detection is to identify behavioral shifts within time series data. This article focuses on scenarios where the data are derived from \changeAgain{an}{a} heterogeneous Poisson process or a marked Poisson process. 
We present a methodology for detecting multiple offline change-points using a minimum contrast estimator. Specifically, we address how to manage the continuous nature of the process given the available discrete observations. Additionally, we select the appropriate number of changes via a cross-validation procedure, which is particularly effective given the characteristics of the Poisson process. Lastly, we show how to use this methodology for self-exciting processes with changes in the intensity. Through experiments, with both simulated and real datasets, we showcase the advantages of the proposed method, which has been implemented in the R package.
\end{abstract}


\paragraph{Keywords.}
Change-point detection;  Point process;  Dynamic programming ;  Cross-Validation


\section{Introduction}

Detecting multiple change-points is one of the most common tasks in statistics, data analysis and signal processing. The problem can be formulated as follows: considering a process observed over time, identify the time instants, called change-points, before and after which the intensity of the process is different. In this work, we consider the detection of change-points in the distribution of a point process. In particular, we focus on the simplest and most fundamental point process, namely the Poisson process (see \cite{ABG93} for a general introduction to Poisson processes), which has been widely used as a reference model in many application domains, including volcanology \citep{bebbington}, epidemiology \citep{epidemio1}, and cyber-attack modeling \citep{cyber1, cyber2}. More precisely, we consider a heterogeneous Poisson process whose intensity function is assumed to be piecewise constant, and our objective is to detect the times at which the intensity changes.

\subsection{State of the art}

\paragraph{General approaches to change-point detection and statistical issues}
Change-point detection has been widely studied in the literature \citep[see, e.g.,][for surveys]{NHZ2016,TOV2020}. Existing methods can be categorized according to several main distinctions: ($a$) single versus multiple change-point detection, ($b$) online versus offline settings, ($c$) frequentist versus Bayesian inference, and ($d$) discrete versus continuous time. All combinations have been considered in one way or another in the literature. The online change-point detection aims to detect changes while the process is being observed, whereas the offline setting addresses the problem of detection once the process has been observed over a fixed period of time. Our work is situated within the framework of offline multiple change-point detection for continuous-time processes in a frequentist setting. 

Within this framework, from a statistical point of view, the detection of change-points raises three main issues: ($i$) locating the change-points, ($ii$) estimating the parameters governing the process between the change-points, and ($iii$) choosing the number of change-points. Problem ($ii$) is generally straightforward to address using, for instance, maximum likelihood estimation, once the change-points have been identified \citep[see, e.g.,][]{YK2001}. Determining the number of change-points ($iii$) is a typical model selection problem, which is examined later. Locating the change-points ($i$) is the most challenging aspect, both in discrete time, because of the combinatorial number of possible segmentations, and in continuous time, as considered here, mainly because it involves minimizing a contrast function (e.g., the negative log-likelihood) that is not continuous with respect to the change-point locations. This problem therefore constitutes the main focus of the present paper.

\paragraph{Change-point detection in Poisson processes}

The online detection of change-points in Poisson processes has recently been considered \citep[see, e.g.,][]{adams2007bayesian, corneck2025online, LXP26}. 
For the same process in the offline setting, Bayesian approaches have been proposed by \cite{RA1986} and \cite{GB2015} for the detection of a single change-point, and by \cite{YK2001} and \cite{S2021} for the detection of multiple change-points. As we have already mentioned, our work falls within the framework of offline multiple frequency-based detection, similar to \cite{grela}, which uses multiple-test techniques, whereas we adopt a segmentation-based approach. 
In the context of daily home activity, \cite{MaK24} consider a hierarchical model that involves heterogeneous Poisson processes; however, the detection of change-points is performed in discrete time.


The optimization problem for determining the optimal locations of change-points (issue ($ii$)) in a discrete time setting is not straightforward because classical contrasts are not continuous with respect to the change-points. The optimization thus requires exploring the whole segmentation space, which grows exponentially with the number of time points. However, the problem can be efficiently solved using dynamic programming  (DP) provided the quantity to be optimized is segment-additive, with a quadratic computational complexity \citep{AuL89}. 
The problem does not become easier in continuous time, because, as we will show, although the change-point locations themselves are continuous parameters, most classical contrast functions are still not continuous with respect to them, and the segmentation space is now infinite-dimensional. 
The exploration of the continuous segmentation space has been considered by \cite{WO1997} and \cite{YK2001} for the detection of a single change-point in a Poisson process. Importantly, \cite{YK2001} observe that the negative log-likelihood is in fact concave between two successive event times, allowing the problem to be solved in a simple and exact manner. This observation is essential and forms the basis of the strategy we propose here for the detection of multiple change-points.
Alternatively, the optimization problem ($ii$) can be tackled by discretizing time (see \cite{AMRPS2008} for medicine, \cite{AREZ2011} for pollution monitoring, or \cite{SZ2012} for genomics), but this comes as the price of computational efficiency and/or accuracy.

\paragraph{Hawkes-type processes}

The present article also allows us to handle a self-exciting process, namely a modified Hawkes process with change-points. Only a few references address change-point detection in Hawkes processes, and most of them rely on online procedures such as \cite{hawkesMLE1, hawkesMLE2}. We also mention \cite{bhaduri} for a sequential testing procedure, and the recent work of \cite{zhang2024conjugate}, who study the problem within a Bayesian framework using Monte Carlo sampling for inference. Changes in Hawkes process intensity are also considered by \cite{WYW23} and \cite{BR2026} in a discrete-time Hawkes setting; this point of view is discussed in the numerical section. Finally, the method proposed by \cite{KBP25} is the closest to ours in the sense that it addresses the detection of multiple change-points in a Hawkes model. Nevertheless, their approach is substantially different, as they study the deterministic equation satisfied by the expectation of the process, and their theoretical results rely on asymptotic regimes where the event rate per unit time becomes very large.

%
%
\subsection{Our contribution}

Let us emphasize the main novelties of this work. 
\begin{description}
\item
[{\it Continuous-time approach without discretization:}] 
 Unlike classic methods that often discretize time or use time "binning" to simplify the problem, this work directly addresses the continuous nature of inhomogeneous Poisson processes. 
\changeAgain{}{While previous continuous-time works have leveraged this concavity for single change-point detection, we theoretically extend this property to a multidimensional segmentation space.} 
We provide a rigorous theoretical justification—based on \changeAgain{}{strict} concavity and additivity of the contrast—to prove that an exact solution can be found without an arbitrary search grid by focusing solely on event times.

\item[{\it Fast and exact optimization via dynamic programming:}] By reducing the continuous search space to a finite grid tied to the observations, the paper enables the use of dynamic programming (DP) \citep{AuL89} to obtain the \changeAgain{}{aforementioned} exact solution \changeAgain{in an efficient manner.}{with an $\mathcal{O}(n^{2}K)$ complexity. }
%
%
\changeAgain{}{Besides, the choice of the contrast is left to the user, provided it satisfies strict concavity and segment-additivity.}

\item[{\it \changeAgain{Introduction of \textit{admissible}}{New} contrasts:}] 
\changeAgain{To avoid degenerate solutions (zero-length segments) frequent with standard maximum likelihood}{Because standard frequentist maximum-likelihood optimization in continuous time typically yields degenerate solutions, such as segments of zero length}, we propose new contrasts, such as the Poisson-Gamma contrast, \changeAgain{which}{in which the Gamma parameters} act as regularization.

\item[{\it Selection of the number of segments via cross-validation:}] The paper proposes a cross-validation procedure made valid by the \textit{thinning property} of Poisson processes\changeAgain{, thus overcoming the usual limitations of cross-validation for discrete segmentation problems}{—a structural advantage that is fundamentally inapplicable in discrete-time segmentation}.

\item[{\it Extension to self-exciting processes (Hawkes-type processes)}:] The methodology is extended to Hawkes-type processes \changeAgain{}{where the intensity changes from one segment to another through a multiplicative coefficient.}
\changeAgain{}{Using }the time-scaling theorem, from \cite{papa}, this allows a complex process to be reduced to a heterogeneous Poisson process to which our methodology applies.

%
\end{description}
%

\subsection{Outline}
We present the model in Section \ref{sec:ModelPP}. Next, the optimization procedure for a given number of segments is described in Section \ref{sec:GeneralCP} and the admissible contrasts are discussed in Section \ref{sec:ContrastPP}. This methodology is first extended to marked Poisson processes in Section \ref{sec:MPP} and to Hawkes(-like) processes in Section \ref{sec:HP}. Section \ref{sec:selection} presents the proposed procedure for selecting the number of segments. Numerical experiments on synthetic data are presented in Section \ref{sec:simu} and the use of the whole methodology is illustrated on earth science datasets in Section \ref{sec:application}.  
The proposed method has been implemented in an R package \texttt{CptPointProcess}, which is available at \url{github.com/Elebarbier/CptPointProcess}. 
Some elements of proof are given in Section \ref{app:proofs}.

\section{Poisson model} \label{sec:ModelPP}

Consider a general Poisson process $\{N_t\}_{0 \leq t \leq 1}$ with intensity $t \mapsto \lambda(t)$ and denote $N_t=N((0,t])$ the number of events occurring over the interval $[0,t]$. Let's denote $T_1, \ldots, T_{N_1}$ the event times observed over the observation time interval $[0,1]$, such that $N_{T_i}-N_{T_i^-}=1$ where ${T_i^-}$ is the time just before ${T_i}$, with the convention $T_0=0$ and $T_{N_1+1}=1$. We assume that $T_i-T_i^-=0$ for all $i$ and designate $N_1=n$ for simplicity. In the following, we work conditionally on this event.

We assume that the intensity of the process is piecewise constant. Let $m$ be a partition of $[0,1]$ composed of $K$ segments denoted $I_k : =(\tau_{k-1}, \tau_k]$ for $k=1, \ldots, K$ where $0=\tau_0< \tau_1 < \ldots < \tau_K=1$ are called the change-points, and $\taubf=(\tau_1,\tau_2, \ldots, \tau_{K-1})$ the vector of $K-1$ unknown change-points.  
Thus, the partition $m$ is defined as $m:= \{I_k\}_{1\leq k \leq K}$ and the associated intensity of the process is  
\begin{equation}\label{eq:intensity}
\lambda(t):= \sum_{k=1}^K \lambda_k \one_{I_k}(t),
\end{equation}
where $\lambda_k$ is the intensity in segment $I_k$.  The vector of intensities is denoted $\lambdabf=(\lambda_1,\lambda_2, \ldots, \lambda_{K})$. This model corresponds to a piecewise homogeneous Poisson process: we call it the change-point Poisson model.

Let $\Delta \tau_k:= \tau_k- \tau_{k-1}$ be the length of the interval $I_k$ and
$\Delta N_k:=N((\tau_{k-1}, \tau_k])=N(\tau_k)-N(\tau_{k-1})$ the number of events occurring in the interval $I_k$. The probability of a given path $\{N_t\}_{0 \leq t \leq 1}$ (which we denote by $N$) for the Poisson change-point model (identified by a subscript $P$) is then as follows
\begin{equation}\label{eq:likelihoodPP}
p_P(N; \taubf, \lambdabf) = \prod_{k=1}^K e^{-\lambda_k \Delta \tau_k} \lambda_k^{\Delta N_k}.
\end{equation}
We can easily see that, for a given $\taubf$, $p_P(N; \taubf, \lambdabf)$ is maximal for $\widehat{\lambda}_k(\taubf) = \Delta N_k / \Delta \tau_k$, $1 \leq k \leq K$. 
Consequently, the objective is to find the set of maximum-likelihood change-point locations that minimize the Poisson contrast:
\begin{equation} \label{eq:contrastPP}
  \gamma_P(\taubf ; N) 
  := - \log p_P(N ; \taubf, \widehat{\lambdabf}(\taubf))
  = \sum_{k=1}^K \Delta N_k \left ( 1-\log \left ( \frac{\Delta N_k}{ \Delta \tau_k} \right) \right).
\end{equation}
Note that $\gamma_P$ is additive with respect to segments, due to the independence of the event times in disjoint segments: this property plays an essential role in the sequel.

\section{Optimal change-points}\label{sec:GeneralCP}
We now consider the determination of the optimal set of change-point locations $\widehat{\taubf}$, defined as the minimizer of a given data-dependent contrast $\gamma(\taubf; N)$, such as the Poisson contrast $\gamma_P$ defined in \eqref{eq:contrastPP}. 
An important feature of the resulting optimization problem on $(0, 1)^{K-1}$ is that the function $\gamma$ may not be convex, or even continuous with respect to $\tau_k$ (the Poisson contrast $\gamma_P$ is neither). Consequently, classical optimization strategies such as gradient descent do not apply. 

Still in the context of single change-point detection, \cite{YK2001} make the critical observation that the Poisson contrast $\gamma_P$ is in fact concave on every inter-event interval $[T_i, T_{i+1}[$, so that the optimal change-point is necessarily located at an event time $T_i$ or just before, in $T_i^-$. 
Extending \cite{YK2001}, we show that, in the context of multiple change-point detection with a given number of segments $K$, some assumptions on the contrast function are sufficient to guarantee that the optimal change-points $\widehat{\tau}_k$ ($1 \leq k \leq K-1$) are all located at, or just before, an event time. This observation reduces the original continuous-space optimization problem to a discrete-space optimization problem, for which efficient and exact algorithms exist.

\subsection{Segmentation space and partitioning} \label{sec:SegSpace}

For a fixed number of segments $K$, we define the segmentation space as the set of all possible partitions of $(0,1)$ into $K$ segments:
$$
\cM^K: = \left\{\taubf = (\tau_1, \ldots, \tau_{K-1}) \in (0,1)^{K-1}; 0 = \tau_0< \tau_1 < \ldots < \tau_K = 1 \right\}.
$$
An element of $\cM^K$ (i.e. a "segmentation") is therefore a sequence of change-points.
In addition, we define the set of all possible $K$-uplets of integers, possibly null, whose sum is $n$:
\begin{equation} \label{space:UKN}
\Upsilon^{K,n} : = 
\left\{\nubf : = (\nu_1, \ldots, \nu_K) \in \{0, 1, \ldots, n\}^K, ~ \sum_{k = 1}^K \nu_k = n\right\}.
\end{equation}
It's a finite set with cardinality $\left| \Upsilon^{K,n} \right| = \binom{n+K-1}{K-1}$.
Then, for a fixed $\nubf \in \Upsilon^{K,n}$ and a given path $N = \{N_t\}_{0 \leq t \leq 1}$ with $n$ events (i.e. $N_1=n$), we define the subset of segmentations from $\cM^K$ with a prescribed number of events in each segment given by $\nubf$, as follows
\begin{equation} \label{space:PnuK}
\cM_\nubf^K(N) : = \left\{\taubf \in \cM^K,  \forall \ 1 \leq k \leq K: ~\Delta N_k = \nu_k \right\}.
\end{equation}
It's important to note that this constraint is equivalent to imposing that, for each $1\leq k \leq K-1$,
$$
\tau_k \in \left [T_{\sum_{j = 1}^k \nu_j}, T_{\sum_{j = 1}^k \nu_j+1} \right )
$$
by setting $T_0 = 0$ and $T_{n+1} = 1$.
It is obvious that the count vectors $\nubf \in \Upsilon^{K,n}$ induce a partition of the segmentation space $\cM^K$: 
$$
\cM^K = \bigcup_{\nubf \in \Upsilon^{K,n}} \cM_\nubf^K (N), 
\qquad
\{\nubf \neq \nubf'\} \Rightarrow \{\cM_\nubf^K(N) \cap \cM_{\nubf'}^K(N)  = \varnothing\}.
$$
Figure \ref{fig:SpaceSegNew} shows the segmentation space $\cM^K$ for $K = 3$ and its partition into all subsets $\cM^K_\nubf (N)$ for $\nubf \in \Upsilon^{K, n}$ with $n = 4$ events. The gray region labeled by $\nubf = (2, 1, 1)$ is $\cM_\nubf^K (N)$, i.e., the set of all segmentations in $(0, 1)^2$ such that $\Delta N_1 = 2, \Delta N_2 = 1, \Delta N_3 = 1$ or, equivalently, $T_2 \leq \tau_1 < T_3$ and $T_3 \leq \tau_2 < T_4$.

\begin{figure}[hbtp]
  \begin{center}
  \includegraphics[trim=10 10 20 0, clip=]{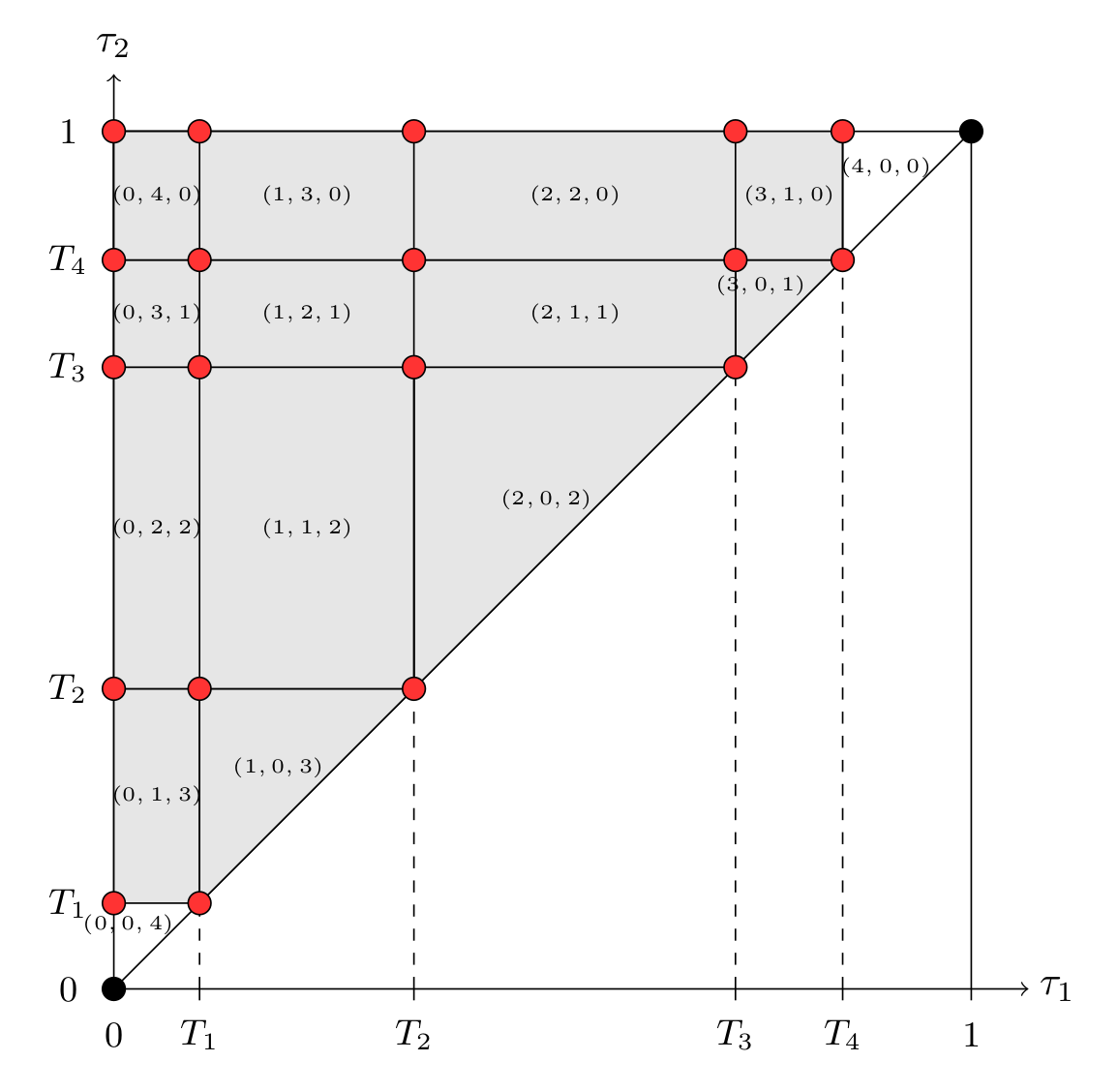}
  \end{center}
  \caption{\change{Segmentation space $\cM^K_\star$ for $K = 3$ and $N_1 = n = 4$ in gray. Each gray block in the upper triangle corresponds to an element of $\Upsilon^{K,n}_\star$. Each white block in the upper triangle corresponds to an element of $\Upsilon^{K,n} \setminus \Upsilon^{K,n}_\star$.}{Segmentation space $\cM^K$ for $K = 3$ and $N_1 = n = 4$ in gray and white in the upper triangle, and its restricted version $\cM^K_\star$ defined in \eqref{RestrictedSegSpace} in gray. Each block corresponds to an element of $\Upsilon^{K,n}$.}} 
  \label{fig:SpaceSegNew}
\end{figure}

\change{
For obvious reasons, we do not consider the configurations $\nubf$ that may contain more than two successive zeros,  i.e.  one of the interval $\displaystyle  [T_{\sum_{j = 1}^{k} \nu_j}, T_{\sum_{j = 1}^{k} \nu_j+1}  )$ contains more than two changes. This may occur only when $K>2$. We therefore restrict the segmentation space to 
\begin{equation} \label{RestrictedSegSpace}
\cM^K_\star: = \bigcup_{\nubf \in \Upsilon^{K,n}_\star} \cM_\nubf^K(N),
\end{equation}
where 
\begin{equation*} \%label{space:UKN_star}
\Upsilon^{K,n}_\star = \left\{ \text{$\nubf \in \Upsilon^{K,n}$, for $2\leq k \leq K-1$: if $\nu_k = 0$ then $\nu_{k-1}\neq 0$ and $\nu_{k+1} \neq 0$} \right\}
\end{equation*}
is of cardinality 
$$
\left| \Upsilon^{K,n}_\star \right| = \sum_{h = 1}^K \binom{n-1}{h-1} \binom{h+1}{K-h},
$$ 
with the convention $\binom{p}{q} = 0$ if $q > p$. In Figure \ref{fig:SpaceSegNew}, the forbidden segmentation subsets (i.e. the elements of $\Upsilon^{K,n} \setminus \Upsilon^{K,n}_\star$) are the white upper triangular regions labeled with $\nubf = (0, 0, 4)$ and $\nubf = (4, 0, 0)$, respectively.
}{}

\subsection{Minimum contrast estimation}  \label{sec:Contrast}

We are now looking for the optimal segmentation (i.e. the sequence of change-points) $\widehat{\taubf}$ with $K$ segments that minimizes a given contrast $\gamma$ within $\cM_{\change{\star}{}}^K$:
\begin{equation} \label{eq:optSeg}
\widehat{\taubf} =  \argmin{\taubf \in \cM^K_{\change{\star}{}}}   \gamma (\taubf).
\end{equation}
As the number of segments $K$ is constant in the rest of the section, the exponent $K$ can be omitted for clarity. \\

The methodology we propose is based on two main assumptions on the contrast function $\gamma$. \\

\begin{ass}[Segment-additivity assumption]\label{hyp:additivity} 
For each $\nubf \in \Upsilon_{\change{\star}{}}^n$ and for all $\taubf \in \cM_\nubf(N)$,  the contrast $\gamma$ writes as sum of segment-specific cost functions:
$$
\gamma (\taubf) = \sum_{k = 1}^K C(\nu_k, \Delta\tau_k).
$$
\end{ass}

As an example, for the Poisson contrast $\gamma_P$ defined in \eqref{eq:contrastPP}, the cost function is $C(\nu_k, \Delta\tau_k) = \nu_k \left  (1-\log \left (\nu_k / \Delta \tau_k\right) \right)$. \\

\begin{ass}[\change{Concavity}{(Strict) concavity} assumption] \label{hyp:concavity} 
For each $\nubf \in \Upsilon_{\change{\star}{}}^n$ and each $1 \leq k \leq K$, the cost function
$C(\nu_k,\Delta \tau_k)$ is a \change{}{(strictly)} concave function of $\Delta \tau_k$.
\end{ass}

Under assumption \ref{hyp:additivity}, by partitioning the segmentation space, the optimization problem \eqref{eq:optSeg} can be rewritten as follows
\begin{equation} \label{eq:est_tau}
\w{\taubf}
= \argmin{\nubf \in  \Upsilon^{K,n}_{\change{\star}{}}}   \min_{\taubf \in \cM_\nubf^K(N)}  \gamma (\taubf) = \argmin{\nubf \in  \Upsilon^{K,n}_{\change{\star}{}}}   \min_{\taubf \in \cM_\nubf^K(N)}  \sum_{k = 1}^K C(\nu_k,\Delta \tau_k).
\end{equation}

We now show that the concavity assumption \ref{hyp:concavity} for each cost function $C$ implies piecewise concavity of the contrast $\gamma$ with respect to $\taubf$. \\

\begin{prop} \label{prop:Concavity_C}
Under assumption \ref{hyp:concavity}, for each $\nubf \in \Upsilon^{K,n}_{\change{\star}{}}$, the contrast function $\taubf \rightarrow \gamma (\taubf)$ is concave with respect to the segmentation $\taubf$ within $\cM_\nubf^K(N)$. \change{}{Under the strict version of assumption \ref{hyp:concavity}, the contrast function $\gamma (\taubf)$ is strictly concave.}
\end{prop}

\change{}{
The proof of this proposition is given in Appendix \ref{sec:proof_concavity}. If this result can be shown easily for a single change-point \citep[see][]{YK2001}, extending it to the multidimensional case, with more than one change-point, is not \change{a straightforward generalization of this result. It relies on studying the negative definiteness of the Hessian matrix of the contrast function.}
{straightforward as neighbor segments have common boundaries, so all change-point locations need to be handled at once. The proof we propose relies on the negative definiteness of the Hessian matrix of the contrast function. }
}

\subsection{Exact optimization}   \label{sec:DP}
We now show that the optimal segmentation necessarily belongs to a finite subset of $\cM^K_{\change{\star}{}}$. 
More precisely, \changeAgain{}{under suitable Assumptions~\ref{hyp:additivity} and~\ref{hyp:concavity}}, it necessarily belongs to a finite and known grid, so that the optimal solution can be obtained in an exact and fast manner.
\change{}{To this aim, we limit ourselves to strictly concave cost and contrast functions.}

\begin{theo} \label{theo:tau}
Under \change{}{the additivity} Assumption~\ref{hyp:additivity}  and \change{}{the strict concavity Assumption} \ref{hyp:concavity},
for a fixed $\nubf \in  \Upsilon^{K,n}_{\change{\star}{}}$ we get 
$$
\w{\taubf} = \argmin{\taubf \in \cM^K_\nubf (N)}  \gamma  (\taubf) \in \{T_{\nu_1}, T_{\nu_1+1}^-\} \times \{T_{\nu_1+ \nu_2}, T_{\nu_1+\nu_2+1}^-\} \times \ldots \times \{T_{\nu_1+\ldots \nu_{K}},T_{\nu_1+\ldots \nu_{K}+1}^-\}  . 
$$
\end{theo}
\begin{proof}
The proof follows directly from the \change{}{strict} concavity of the contrast function $\gamma (\taubf)$ with respect to $\taubf \in \cM^K_\nubf (N)$ (see Proposition \ref{prop:Concavity_C}).
\end{proof}

Theorem \ref{theo:tau} guarantees that the $K-1$ optimal change-points in each subset of $\cM^K_\nubf (N)$ are necessarily one of its boundary partitions, thus reducing the search to a finite set of possible solutions. These solutions are illustrated in Figure \ref{fig:SpaceSegNew} by the red circles for the simple case of $K=3$ segments (or $2$ change-points). For example, for $\nubf = (2,1,1)$, $T_2 \leq \tau_1 < T_3$ and $T_3 \leq \tau_2 < T_4$ leading to four possible choices for the optimal solution: $(\w{\tau}_1,\w{\tau}_2) = (T_2,T_3), (T_2,T^-_4),(T_3^-,T_3)$ or $(T_3^-,T^-_4)$.

\changeAgain{}{Theorem \ref{theo:tau} ensures that, under Assumptions \ref{hyp:additivity} and \ref{hyp:concavity}, the global minimizer of the contrast can be recovered with quadratic complexity; it does not state that this complexity is itself optimal (see Discussion).}

\change{}{
Note that Theorem \ref{theo:tau} allows configurations $\nubf$ that may contain more than two consecutive zeros,  i.e.,  one of the intervals $\displaystyle  [T_{\sum_{j = 1}^{k} \nu_j}, T_{\sum_{j = 1}^{k} \nu_j+1}  )$ containing more than two change-points (when $K>2$). 
Due to the strict concavity assumption, the optimal location of the two change-points will be either in $T_{\sum_{j = 1}^{k} \nu_j}$ or in $(T_{\sum_{j = 1}^{k} \nu_j+1})^-$, which will give a segment of zero length associated with a zero count.
One may further observe that, if $C(0, 0) = 0$, an arbitrary number of such segments can be added, without changing the contrast function. 
\changeAgain{}{As a consequence, a segmentation with $K$ segments including an empty segment with null length corresponds to a segmentation with $K-1$ segments, both yielding the same contrast value. We consider that such \changeAgain{}{a} segmentation is actually a segmentation with $K-1$ segments, and not with $K$ segments.} 
We therefore restrict the segmentation space to 
\begin{equation} \label{RestrictedSegSpace}
\cM^K_\star: = \bigcup_{\nubf \in \Upsilon^{K,n}_\star} \cM_\nubf^K(N),
\end{equation}
where 
\begin{equation*} 
\Upsilon^{K,n}_\star = \left\{ \text{$\nubf \in \Upsilon^{K,n}$, for $2\leq k \leq K-1$: if $\nu_k = 0$ then $\nu_{k-1}\neq 0$ and $\nu_{k+1} \neq 0$} \right\}
\end{equation*}
is of cardinality 
$$
\left| \Upsilon^{K,n}_\star \right| = \sum_{h = 1}^K \binom{n-1}{h-1} \binom{h+1}{K-h},
$$ 
with the convention $\binom{p}{q} = 0$ if $q > p$. In Figure \ref{fig:SpaceSegNew}, the forbidden segmentation subsets (i.e. the elements of $\Upsilon^{K,n} \setminus \Upsilon^{K,n}_\star$) are the white upper triangular regions labeled with $\nubf = (0, 0, 4)$ and $\nubf = (4, 0, 0)$, respectively.
}

Consequently, the optimal change-points in $\cM^K_\star$ are necessarily located at or just before an event $T_i$ or $T_i^-$. 
\change{Moreover, for obvious reasons, for $K>2$, the configuration $(\nu_k,\Delta {\tau}_k) = (0,0)$ is not taken into account. 
This configuration corresponds to cases where $\tau_{k-1}$ and $\tau_{k}$ belong to the same inter-event interval $\left [T_{\sum_{j = 1}^{k-1} \nu_j}, T_{\sum_{j = 1}^{k-1} \nu_j+1} \right )$ and $\tau_{k-1} = \tau_{k}$.}{}
The global optimization problem is thus reduced to a discrete optimization problem on the finite grid
$$
\tbfp = \{T_1^-, T_1, , T_2^-, T_2, \ldots ,T_n^-,T_n\},
$$
which has cardinal $2 n$. This problem is well known in the context of discrete change-point detection and can be solved using the classical dynamic programming (DP) algorithm. This algorithm and its principle are presented in \ref{app:DP} in a general context of discrete time detection, and we then explain how to apply it on a fixed grid. In our case, the number of points of the grid is $2n$ so the complexity of the algorithm is $\cO(n^2 K)$. This algorithm can be used here thanks to the segment-additivity assumption \ref{hyp:additivity} of the contrast function $\gamma$.

\section{Proposed admissible contrasts}\label{sec:ContrastPP}
In this section, we present two admissible contrasts in the sense that they satisfy the two assumptions \ref{hyp:additivity} and \ref{hyp:concavity} and discuss the relevance of their optimal solutions.

First\changeAgain{}{,} let \changeAgain{}{us} consider the classical negative $\log$-likelihood Poisson contrast which is written, for a given $N$ path of the Poisson change-point model, a fixed $\nubf \in \Upsilon^{K,n}_\star$ and any $\taubf \in \mathcal{P}_\nubf^K(N)$, as follows
\begin{equation} \label{contrast:P}
 \gamma_P (\taubf)
 = -\log p_P(N ; \taubf, \w{\lambdabf}) = \sum_{k=1}^K  \nu_k \left  (  1-\log \left ( \frac{\nu_k}{  \Delta \tau_k} \right) \right)
 = \sum_{k=1}^K  C(\nu_k,\Delta \tau_k).
 \end{equation}
The independence property between disjoint time intervals in the Poisson process guarantees that the additivity Assumption \ref{hyp:additivity} is verified. 
It is then easy to see that the cost function $C$ is concave with respect to $\Delta \tau_k$, and that Assumption \ref{hyp:concavity} is valid. \change{}{However, the strict version of the concavity Assumption \ref{hyp:concavity} does not hold when $\nu_k = 0$.}

\change{
However, the Poisson contrast \eqref{contrast:P}, as the one based on least-squares, see \ref{app:contrasts}, inevitably leads to an optimal solution for $K>2$ that contains zero-length segments. The reason is easy to understand: the cost $C(1,0)= - \infty$. For obvious reasons, such solutions are undesirable. 
}{
Furthermore, because $C(1,0)= - \infty$, as soon as $K > 2$, this contrast inevitably yields segments with null length and count one, that is segments with boundaries located on both sides of a single event: $(T_i^-, T_i]$. This also holds for the least-squares contrast (see \ref{app:contrasts}). Such solutions are clearly undesirable. 
}

\change{}{
One possible solution would be to impose the constraint $\Delta \tau_k > \epsilon$ for a certain $\epsilon > 0$. 
However, this raises two problems.
First, we must choose an appropriate value for $\epsilon$.
Second, if $\epsilon$ is sufficiently small, the solution will always lie on the boundary of each subset, simply shifted by $\epsilon$. 
To get around this problem, we propose adopting a Bayesian-type contrast in which hyper-parameters play the role of regularization parameters. 
The problem of choosing $\epsilon$ is thus transferred to the choice of hyper-parameters, but this approach allows us to obtain the exact solution that is not located at the boundaries.  
}
We assume that the intensities $\lambda_k$ are independent random variables and follow a Gamma distribution with parameters $a>0$ and $b>0$. For a fixed $\nubf \in \Upsilon^{K,n}_\star$ and all $\taubf \in \mathcal{P}_\nubf^K(N)$, the proposed contrast, the so-called Poisson-Gamma contrast, is
\begin{align} \label{contrast:PG}
  \gamma_{PG} (\taubf)
  & =- \log{p_{PG}(N ; \taubf; a, b)}= - \log{ \int  p_{PG}(N, \lambdabf ; \taubf, a, b)   \d \lambdabf}, \\
  & = \sum_{k=1}^K \left (-a \log{b} +\log{\Gamma(a)}  +\widetilde{a}_k \log{\widetilde{b}_k} - \log{\Gamma(\widetilde{a}_k)} \right)=\sum_{k=1}^K C(\nu_k,\Delta \tau_k), \nonumber 
 \end{align}
where $\widetilde{a}_k =  \nu_k+a$, $\widetilde{b}_k =  \Delta \tau_k+b$. The details of the derivation of this contrast are given in \ref{app:contrasts}. 
\change{This contrast also satisfies assumptions \ref{hyp:additivity} and \ref{hyp:concavity} and is therefore admissible.}{This contrast satisfies the additivity Assumption \ref{hyp:additivity} and the strict concavity Assumption \ref{hyp:concavity} and is therefore admissible.}
\change{In addition}{Indeed}, $C(1,\Delta \tau_k)$ is now lower-bounded, allowing another segmentation to be preferred to the undesirable previous ones. 

Note that $a$ and $b$ must be chosen in practice.  Since the mean of the Gamma distribution with parameters $(a,b)$ is $a/b$, a simple rule is to choose $a$ and $b$ so that $a/b=n$. \changeAgain{}{By default, we propose to take $a=1$ and $b=1/n$. The simulation study presented in~\ref{app:appSimuls} shows that this combination works well in practice in terms of model selection, accuracy of the change-point location and estimation of the cumulated intensity, see Figure~\ref{fig:simPP-aRobust}. }

Moreover, the conditional distribution of $(\lambdabf \mid N;\taubf)$ is a Gamma distribution with parameters $\widetilde{a}_k$ and $\widetilde{b}_k$, thus its posterior mean is
\begin{equation} \label{PosteriorMeanLambda}
  \E(\lambda_k \mid N, \tau) = \widetilde{a}_k / \widetilde{b}_k.
\end{equation}
We thus obtain the estimator $\widehat{\lambda}_k = \widetilde{a}_k / \widetilde{b}_k$ and we can construct the estimator of the intensity process $(\lambda(t))_{t \in [0,1]}$. We will use this estimator in the final algorithm (see Section \ref{sec:selection}).

\section{Extension to marked Poisson process} \label{sec:MPP}
We extend the proposed methodology to a marked Poisson process for which both the intensity function of the underlying Poisson process and the parameter of the distribution of the marks are affected by the same changes.

\subsection{Model} \label{sec:model:MPP}

Let us consider a marked Poisson process. More specifically, we consider a Poisson process $N$ with a piecewise constant intensity function $\lambda$ given in Equation \eqref{eq:intensity} and suppose that a mark $X_i$ is associated with each event time $T_i$  ($1 \leq i \leq n = N_1$).  \\
We assume that the marks $\{X_i\}_{i=1,\ldots,n}$ are independent exponential random variables with parameter $\rho(T_i)$: $X_i \mid T_i \sim \mathcal{E}(\rho(T_i))$. The function $\rho$ is assumed to be piecewise constant with the same change-points as $\lambda$, for $t \in [0, 1]$:
\begin{equation} \label{eq:rho}
\rho(t)= \sum_{k=1}^K \rho_k \one_{I_k}(t), \quad \rhobf=(\rho_1, \ldots, \rho_K).
\end{equation}

\subsection{Admissible contrasts}

The $\log$-likelihood of a given marked Poisson path $(N, X)$ with piecewise intensity and piecewise mark distribution is given by
 \begin{align} \label{eq:logLikMarkedPoisson}
  & \log p_{MP}(N, X ; \taubf, \lambdabf, {\rhobf}) \\
  & \quad = \sum_{k=1}^K \left ( \Delta N_k \log(\lambda_k) - \lambda_k \Delta \tau_k+  \sum_{i, T_i \in I_k} \log p(X_i \mid T_i; \rho_k) \right). \nonumber
\end{align}
 It can be easily seen that, for a given $\taubf$, $\log p_{MP}(N,X ; \taubf, \lambdabf, {\rhobf})$ is maximal for $\widehat{\lambda}_k(\taubf) = {\Delta N_k}/ {\Delta \tau_k}$ and 
 $$
 \w\rho_k (\taubf)=  \frac{ \Delta N_k}{S_k}, 
 \qquad \text{with} \quad S_k = \sum_{i, T_i \in I_k} X_i,
 $$
 for $1 \leq k \leq K$. 
 The resulting contrast function for the estimation of the change-points can be written, for a fixed $\nubf \in \Upsilon^{K,n}_\star$ and for all $\taubf \in \mathcal{P}_\nubf^K(N)$, as 
\begin{align} \label{eq:contrastMPP}
 & \gamma_{MP} (\taubf) \\
 & \quad = -\log p_{MP}(N, X ; \taubf, \w{\lambdabf},\widehat{\rhobf}) =  \sum_{k=1}^K  \nu_k \left ( -\log{\left ( \frac{\nu_k}{\Delta \tau_k}\right )}-\log{\left ( \frac{\nu_k}{S_k}\right )}+2 \right). \nonumber
 \end{align}
As for the Poisson contrast,  this likelihood-based contrast is admissible but suffers from the limitation pointed out in the previous section. 

Following the same idea, we define a new Marked Poisson-Gamma-Exponential-Gamma (MPGEG) likelihood-based contrast.  Precisely, we assume that the $\lambda_k$'s and $\rho_k$'s are independent random variables and follow a Gamma distribution with parameters $a_\lambda>0$ and $b_\lambda>0$ and a Gamma distribution with parameters $a_\rho>0$ and $b_\rho>0$, respectively. For a fixed $\nubf \in \Upsilon^{K,n}_\star$ and for any $\taubf \in \mathcal{P}_\nubf^K(N)$, the MPGEG contrast is
\begin{align} \label{contrast:PGEG}
  \gamma_{MPGEG} (\taubf)
  & =-  \log{p_{MPGEG}(N ; \taubf , a_\lambda, b_\lambda, a_\rho, b_\rho)}, \nonumber \\
  & = -  \log{\iint  p_{PGEG}(N, \lambdabf,\rhobf ; \taubf , a_\lambda, b_\lambda, a_\rho, b_\rho) \  \d \lambdabf \  \d \rhobf}, \nonumber\\
  & =  \sum_{k=1}^K \left ( \widetilde{a}_{k, \lambda} \log{\widetilde{b}_{k, \lambda}} - \log{\Gamma(\widetilde{a}_{k, \lambda})} \right) +\left (  \widetilde{a}_{k, \rho} \log{\widetilde{b}_{k, \rho}} - \log{\Gamma(\widetilde{a}_{k, \rho})} \right) \nonumber\\
  & + K \left ( -a_\lambda \log{b_\lambda } +\log{\Gamma(a_\lambda )}  -a_\rho \log{b_\rho } +\log{\Gamma(a_\rho )} \right ), 
\end{align}
where $\widetilde{a}_{k, \lambda} =  \nu_k+a_\lambda$, $\widetilde{b}_{k, \lambda} =  \Delta \tau_k+b_\lambda$, $\widetilde{a}_{k, \rho} =  \nu_k+a_\rho$, and $\widetilde{b}_{k, \rho} = S_k+b_\rho$. The derivation of this contrast is detailed in \ref{app:contrasts}. This contrast function is also admissible and can avoid segmentations with zero-length segments.

The posterior means of $\lambda_k$ and $\rho_k$ are 
$$
  \widehat{\lambda}_k  = \widetilde{a}_{k,\lambda} / \widetilde{b}_{k,\lambda}, 
  \widehat{\rho}_k  = \widetilde{a}_{k,\rho} / \widetilde{b}_{k,\rho}.
$$
We use these posterior means as estimators of the intensity of the process and of the density parameter of the marks, respectively.

Similar to the PG contrast, the parameters $a_\lambda, b_\lambda, a_\rho, b_\rho$ need to be chosen. For the two first parameters,  we choose $a_\lambda=1$ and $a_\lambda/b_\lambda = n$. Then, because the conditional distribution of $X \mid N$ is a Pareto distribution with parameter $(b_\rho, a_\rho)$ so $\E[X \mid N]= b_\rho/(a_\rho-1)$ with the condition $a_\rho>2$ (to ensure the existence of the variance). Thus, one may choose $a_\rho = 2.01$ and $b_\rho= \overline{X}(a_\rho-1)$.

\change{}{
\begin{remark}
We choose here the exponential distribution as an example: \changeAgain{this setting work can easily be extended}{this setting can easily be extended} to other distributions for the marks.
In particular, one could consider categorical discrete random variables. In this case, the previous log-likelihood Formula \eqref{eq:logLikMarkedPoisson} still holds, but the optimal parameters for the marks must be appropriately computed.   
\end{remark}
}

\section{Extension to self-exciting process}\label{sec:HP}
In this section, we show that the segmentation of a particular Hawkes-type process can be rewritten as the segmentation of a Poisson process after an adapted transformation.

\subsection{Model} 

We consider a self-exciting process $N$ with conditional intensity process denoted $\lambda(t)$ at time $t$, depending on the entire history before time $t$. The compensator or density of this process is represented by $\Lambda(t)$. As in the Poisson model, we denote the observed event times within the observation interval $[0,1]$ as $(T_i)_{i \geq 1}$.
We further assume that the intensity in each segment of the partition $\{I_k=[\tau_{k-1}, \tau_k)\}_{1\leq k \leq K}$ is the product of a linear exponential Hawkes intensity $\lambda_0(t)$ and a segment-specific constant. Specifically, the conditional intensity function is defined as follows:
\begin{equation}\label{eq:cond_int_php}
    \lambda(t)
     = \sum_{k=1}^K c_k \one_{[\tau_{k-1}, \tau_k)}(t) \lambda_0(t),
    \text{with } \lambda_0(t) = \left(1 + \sum_{T_i < t } \alpha e^{-\beta(t-T_i)}\right),
\end{equation}
where $(c_k)_{1\leq k \leq K}$ is a sequence of positive constants, and the parameters $\alpha$ and $\beta$ are positive constants, referred to as the increase rate and the decay rate respectively. For identifiability reasons, the immigration rate of the baseline intensity $\lambda_0$ is set to 1.
It gives the conditional risk of the occurrence of an event at $t$ given the realization of the process over $[0,t)$.
For the classical self-exciting Hawkes process ($c_k=1$ for all $k$), each arrival in the system instantaneously increases the arrival intensity by $\alpha$, then over time this arrival's influence decays at rate $\beta$. 
Here both the immigration rate and the influence of each event {on the events to come depend on the segment of arrival. \changeAgain{}{This modeling assumption is crucial, as the entire procedure relies on it. At this stage, the proposed method does not generalize to other intensity structures.} 
The model \eqref{eq:cond_int_php} can describe the occurrences of self-exciting events, subject to switches of an underlying regime, such as neuronal activity or crime activity, for example. 
We refer to this model as Piecewise-Hawkes-Process (PHP). 
  
We assume that the process $N$ satisfies the following assumption. 
\begin{ass}[Non-explosion assumption]\label{hyp:non-explosion} 
    The parameters of intensity of the process defined in \eqref{eq:cond_int_php} satisfy
    $$
    \frac{\alpha}{\beta} \max_{k} \{c_k\} < 1.
    $$ 
\end{ass}    
The condition comes from the fact that a standard Hawkes process with conditional intensity $c_{\max}\lambda_{0}$ doesn't explode if and only if   $\max_{k} c_k {\alpha}/{\beta} < 1$. Hence, as 
$$
\lambda(t) \leq \max_{k} c_k\lambda_{0}(t),
$$ 
Assumption \ref{hyp:non-explosion} is sufficient to ensure that $\E[N_t]< \infty$ and $\E[ \lambda(t)]< \infty.$
Some details about the simulation algorithm through thinning are given in \ref{app:simuHP}.
Let us denote for the Hawkes process $N_0$ with conditional intensity $\lambda_0$(t), the
compensator $\lambda_0$, given by the following equation, 
\begin{equation}\label{eq:transformation_compensator}
    \Lambda_0(t) := \int_{0}^{t} \lambda_0(s) ~\d s = t 
    + 
    \frac{\alpha}{\beta} 
    \sum_{T_i <t} 
    \left(
        1- e^{-\beta (t-T_i)}
    \right).
\end{equation}
The link between the process $N$, defined by \eqref{eq:cond_int_php}, and $N_0$ is given in Proposition \ref{prop:link_compensators}.
\begin{prop} \label{prop:link_compensators}
For any $k \in \{1, \dots, K\}$ and for $t \in [\tau_{k}, \tau_{k+1})$, we have $\lambda(t) = c_k \lambda_0(t)$ and
\begin{align}\label{eq:compensator_exponential_hawkes_like}
\Lambda(t) &= 
        c_{k+1} 
        \left[ 
            \Lambda_{0}(t) - \Lambda_{0}(\tau_{k})
        \right]
        +
        \sum_{j=1}^{k} 
        c_j 
        \left[ 
            \Lambda_{0}(\tau_j) - \Lambda_{0}(\tau_{j-1})
        \right].
    \end{align}
\end{prop}
We are now able to define a method of change-point detection in this model. 

\subsection{Change-point detection strategy}

We now describe a strategy to detect the change-points in PHP, which consists in transforming it back into a Poisson process, to which our methodology applies. 

\paragraph{Conversion to a piecewise Poisson process} 
The transformation we employ is based on a well-established result for point processes: the random time-change or time-scaling theorem, originally presented in \cite{papa} and discussed in works such as \cite{daleyII}.
Theorem \ref{th:Modified_time_rescaling} provides an adapted version of this theorem specifically for the piecewise version of the Hawkes process considered here.
\begin{theo}[Modified time-rescaling]
\label{th:Modified_time_rescaling}
Let $(T_i)_{i= 1}^N$ be the event times observed over $[0,1]$ of a PHP  with a conditional intensity $\lambda$ defined by Equation \eqref{eq:cond_int_php}, and satisfying assumption~\ref{hyp:non-explosion}. Then, the sequence $(\Lambda_0(T_i))_{i}$ is a Poisson process with a piecewise constant intensity function $\lambda_{P}$ defined as
\begin{equation} \label{eq:transformed_poisson_intensity}
    \lambda_{P} (t) := \sum_{k=1}^K c_k \mathds{1}_{[\Lambda_0(\tau_{k-1}), \Lambda_0(\tau_{k}))}(t).
 \end{equation}
Furthermore, the likelihood of the exponential Hawkes-type process is
 \begin{equation} \label{eq:transformed_poisson_loglik}
    p_H(N;\taubf, \cbf, \alpha,\beta ) 
    = \left(\prod_{i=1}^n \lambda_0(T_i) \right)
    \left(\prod_{k=1}^K c_k^{n_k} e^{-c_k (\Lambda_0(\tau_k) - \Lambda_0(\tau_{k-1}))}\right).
\end{equation}
\end{theo}
The proof is relegated in Section \ref{app:proofs}. 

For this result, the shape of the conditional intensity of the jump process that we have considered in Equation~\eqref{eq:cond_int_php} is crucial. \changeAgain{}{Indeed, the fact that the parameters all change through a multiplicative factor from one segment to the other is a key argument in proving Theorem \ref{th:Modified_time_rescaling} (see Equation \eqref{eq:joint_density_events1} in Section \ref{th:Modified_time_rescaling}), and thus justifying the procedure.} 

Let us also notice that, because $\Lambda_0$ is a strictly increasing function \citep[see][Definition 4]{LaubAL2015}, the ordering of the event times remains unchanged before and after transformation. Consequently, 
for each $k \in \{0, \dots, K-1\}$, the number of events for the original Hawkes-type process within $[\tau_k, \tau_{k+1})$ is the same as the number of events for its Poisson version within $[\Lambda_0(\tau_k), \Lambda_0(\tau_{k+1}))$. Let us finally mention that the transformed process is not in $[0,1]$ anymore but in the random time interval $[0, \Lambda_0(T_{N[0,1]})]$.
Thus, using Theorem \ref{th:Modified_time_rescaling}, we transfer the task of determining the change-points $(\tau_k)_k$ from $N$ to a change-point detection problem on the Poisson process $(\Lambda_0(T_i))_{i}$ on which the proposed method can be applied. Because the first bracket of the likelihood \eqref{eq:transformed_poisson_loglik} does not depend on the $\tau_k$'s nor on the $c_k$'s, the maximization with respect to the $\tau_k$ of the obtained Poisson likelihood
$$
\prod_{k=1}^K c_k^{n_k} e^{-c_k (\Lambda_0(\tau_k) - \Lambda_0(\tau_{k-1}))}
$$
is equivalent to the maximization of the Hawkes-like likelihood \eqref{eq:transformed_poisson_loglik}, for fixed $\alpha$ and $\beta$.
\change{}{Besides,} as seen in Section \ref{sec:ContrastPP}, the optimal segmentation according to the Poisson contrast of the transformed process may contain segments with null length, so, in practice, we will use the Poisson-Gamma contrast in place of the Poisson contrast.

\paragraph{Estimation of $\alpha$ and $\beta$}
The transformation of an exponential PHP into a heterogeneous Poisson process requires the knowledge of the parameters $\alpha$ and $\beta$. \changeNew{We describe here two possible strategies.}{We distinguish two situations.} \changeNew{First, we may assume that}{In the first situation,} the process goes through a homogeneous period, that is a period $[0, T_{\rm learn}]$ with no change-point. This is the case in some neuronal analyses \citep[see, e.g.][]{stefano2023heterogeneous}, where the spiking activity of a mouse is recorded during activity but also during sleep phases, which can therefore be considered as homogeneous.
\changeNew{We}{In this case, we} propose to estimate the parameters $\alpha$ and $\beta$ on this homogeneous period via, for example, maximum likelihood, and to plug them into the compensator~\eqref{eq:transformation_compensator}; denoting $\w{\theta}= (\w{\alpha}, \w{\beta})$, we thus define
$$
\w{\Lambda}_0(t)
:= \Lambda_{0,\w{\theta}}(t)
= 1 + \frac{\w\alpha}{\w\beta} \sum_{T_i <t} \left(1- e^{-\w\beta (t-T_i)}\right).
$$
Algorithm~\ref{algo:change_pts_detection_php} gives the pseudo-code of the PHP change-points detection algorithm (CDPHP) following this strategy.

\begin{algorithm}[H]
    \caption{change-points detection for PHP (CDPHP)}
    \label{algo:change_pts_detection_php}
    \begin{algorithmic}[1]
      \State{\textbf{input} 
      \begin{itemize}
          \item $(S_k)_{k=1}^{N_0([0,T_{\rm learn}])}$ a sample from an HP with conditional intensity $\lambda_0$ as defined in Equation~\eqref{eq:cond_int_php}
          \item $(T_i)_{i=1}^{N([0,1])}$ a sample from a PHP of conditional intensity $\lambda$ as defined in Equation~\eqref{eq:cond_int_php}
      \end{itemize}}
      \State \textbf{estimate} $\theta$ the parameters of $\lambda_{0}$ based on $(S_k)_{k=1}^{N_0([0,T_{\rm learn}])}$
      \State \textbf{compute} $(\w{\Lambda}_{0}(T_i))_{i=1, \ldots, N([0,1])}$ (a sequence on $[0,\w{\Lambda}_{0}(T_N)]$)
      \State \textbf{find} the set of the $K$ indexes $I_K^N$ of the change-points in the sample $(\w{\Lambda}_{0}(T_i))_{i=1}^N$ using a change-point detection method for an inhomogeneous-PP (re-scaled on $[0,1]$)
     \State {\textbf{return } $\w{\tau}_{1}, \ldots, \w{\tau}_K=(T_j)_{j \in I_{K}^N}$ }
    \end{algorithmic}
\end{algorithm}

\changeNew{Alternatively, when}{The second situation occurs when} no homogeneous period is available. \changeNew{}{Then} we propose to fix $\beta$ and to run the change-point detection method with a grid of values of $\alpha$, keeping the value of $\alpha$ that maximizes the likelihood function. Several values for $\beta$ can then be investigated, even if its value is less sensitive. We use this latter strategy in the application presented in Section \ref{sec:appHP}. 
\change{}{This procedure is natural if only one sample is available. Moreover, one could consider first estimating these parameters over the whole trajectory, assuming initially that there is no change-point. Nevertheless, a preliminary empirical study, not shown, indicates that this approximation is too rough and biases the procedure too much.}

\begin{remark}[Other piecewise point processes.]
    The strategy presented in this section, not only applies to exponential Hawkes processes.
    Indeed, considering a point process with conditional intensity $\lambda_0(t)$, any piecewise version of it with conditional intensity $\lambda(t)$ with the form given in Equation \eqref{eq:cond_int_php} can be segmented with the proposed approach, due to the time-scaling theorem.
\end{remark}

\section{Choosing the number of segments}\label{sec:selection}

Now, let us discuss the selection of the number of segments $K$. To address this, we suggest employing a cross-validation strategy that leverages a property of the Poisson process, which we outline for the reader below.

\begin{propri}[Thinning] \label{prop:thinning}
  Consider a heterogeneous Poisson process $N$ with intensity function $\lambda(t)$: $N = \{N_t\}_{0 \leq t \leq 1} \sim PP(\lambda)$. Sampling each event time of $N$  with probability $f$ results in a heterogeneous Poisson process $N^A$ with intensity function $\lambda^A(t) = f \lambda(t)$. Furthermore, the remaining fraction of event times forms a second heterogeneous Poisson process $N^B$ with intensity function $\lambda^B(t) = (1-f) \lambda(t)$, and the processes $N^A$ and $N^B$ are independent. \\
\end{propri}

This thinning property has two important consequences. First, if the intensity function $\lambda$ is piecewise constant, then the intensity functions $\lambda^A$ and $\lambda^B$ are also piecewise constant, {\sl with the same change-points as $\lambda$}. Second, whatever the form of the intensity function of $N$, the ratio between the intensity functions of $N^B$ and $N^A$ is constant and equal to $\lambda^B(t)/\lambda^A(t) \equiv (1-f)/f$. This suggests the following cross-validation procedure: ($i$) sample events from the observed process $N$ with probability $f$ to form a {\sl learning process $N^L$} and form an independent {\sl test process $N^T$} with the remaining events; ($ii$) for a series of values of $K$, get estimates $(\widehat{\taubf}^{K, L}, \widehat{\lambdabf}^{K, L})$ of the change-points and intensities, respectively, using the learning process $N^L$, ($iii$) evaluate the contrast on the test process $N^T$ with parameters $(\widehat{\taubf}^{K, L}, \frac{1-f}{f} \widehat{\lambdabf}^{K, L})$. 
The whole cross-validation procedure is summarized in \ref{app:cv}, Algorithm \ref{algo:cv}.

\paragraph{Some comments} \change{Few comments}{A few comments} can be made about this procedure.
\begin{itemize}
\item First, Algorithm \ref{algo:cv} involves three tuning parameters: the number of cross-validation samples $M$, the maximum number of segments $K_{\max}$ and the sampling probability $f$. The first two are limited only by the computational burden and can be as large as desired. The sampling probability $f$ was set to $4/5$ in this study. Figure \ref{fig:simPP-vRobust} in \ref{app:appSimuls} shows that this choice yields good performances, as compared to other typical ones.
\item Second, we use the Poisson-Gamma contrast to estimate the parameters $\taubf$ and $\lambdabf$ because it is admissible as explained in Section \ref{sec:ContrastPP}. Still, because the undesirable properties of the Poisson likelihood have no effect when used to measure the fit of $\widehat{\taubf}$ and $\widehat{\lambdabf}$ parameters to an independent process, we use the standard Poisson contrast to evaluate this fit. 
\item \changeAgain{}{Algorithm~\ref{algo:cv} uses the Poisson-Gamma contrast $\gamma_{PG}$ for the learning step, and the Poisson contrast $\gamma_P$ for the test step: we regularize the loss function during training for stability reasons, but we evaluate the model’s generalization ability using the genuine Poisson log-likelihood.}
\item Finally, but most importantly, cross-validation is generally not applicable when dealing with discrete time segmentation problems. Indeed, in this case, it consists in eliminating observation times, which makes the location of the estimated change-point unclear with respect to the complete dataset. The situation is different in the continuous time setting we consider, thanks to the thinning Property~\ref{prop:thinning}.
\end{itemize}

\paragraph{Practical implementation}
In practice, to perform the segmentation of an observed process $N$, we propose to first determine $\widehat{K}$ using the algorithm \ref{algo:cv}, then to estimate $\taubf$ using the Poisson-gamma contrast and $\lambdabf$ as the posterior mean on the whole process $N$:
\begin{align*}
  \widehat{\taubf} & =  \argmin{\nubf \in  \Upsilon^{\w{K},n}_\star}   \min_{\taubf \in \cM_\nubf^{\w{K}}(N)}   \gamma_{PG} (\taubf), \\
  \widehat{\lambda}_k &  = \E(\lambda_k \mid N, \widehat{\taubf}) \quad \text{for all } k = 1, \dots, \widehat{K}.
\end{align*}

The adaptation of the algorithm \ref{algo:cv} to the segmentation of a marked Poisson process is straightforward, replacing the Poisson-Gamma contrast with the marked Poisson-Gamma-Expo-Gamma contrast and the Poisson contrast with the marked Poisson contrast, respectively.

\section{Simulation study} \label{sec:simu}
We present in this section a simulation study to evaluate the performances of the proposed methodology.
In Section \ref{sec:simudesign} we explain the investigated cases, and in Section \ref{sec:simuresults} we give the result when the simulation is done under the Poisson model, the marked-Poisson model or the Hawkes model (PHP). 

\subsection{Simulation design, quality criteria \changeNew{}{and benchmark}}\label{sec:simudesign}


\paragraph{Poisson process} We use a simulation design for the change-points inspired \changeAgain{from}{by} \cite{Chakar2017}. We set the number of segments to $K = 6$, with $\taubf = [0, 7, 8, 14, 16, 20, 24]/24$ change-point locations, i.e. with lengths $\Delta\taubf = [7, 1, 6, 2, 4, 4]/24$. The total length of odd segments ($k = 1, 3, 5$) is $\Delta\tau_- = 17/24$ and the total length of even segments is $\Delta\tau_+ = 7/24$. The intensity is set to $\lambda_-$ in odd segments and to $\lambda_+$ in even segments, based on the two following parameters:
\begin{description}
 \item[$\lambdabar$]: mean intensity: $\lambdabar = \int_0^1 \lambda(t) \d t = (\lambda_- \Delta\tau_- + \lambda_+ \Delta\tau_+)$. It controls the expected number of events. A small value for $\lambdabar$ yields a very scarce signal, i.e. a very poor available information. We consider the values $\lambdabar =$ 32, 56, 100, 178, 316, 562 and 1000.
 \item[$\lambda_R$]: ratio between even and odd intensities: $\lambda_R := \lambda_+ / \lambda_- \geq 1$. It controls the contrasts between successive segments. Note that $\lambda_R = 1$ actually yields a single segment with intensity $\lambdabar$. We consider the values $\lambda_R=$ 1, 2, 3, 4, 6, 8, 11 and 16. 
\end{description}
The values of the intensities $\lambda_-$ and $\lambda_+$ are then $\lambda_- =\lambdabar / \left(\Delta\tau_- + \lambda_R\Delta\tau_+\right)$ and $\lambda_+ =\lambda_R \lambda_-$. 
Examples of such piecewise intensity functions $\lambda$ are given in \ref{app:simulation}, Figure \ref{fig:lambda} for $\bar{\lambda}=100$ and three values of $\lambda_R$. The values of $\lambda_-$ and $\lambda_+$ deduced for all the considered values of the mean intensity $\lambdabar$ and the ratio $\lambda_R$ are given in Table \ref{table:lp_lm} in \ref{app:simulation}, together with a measure, called SNR, we defined to evaluate the difficulty of the detection problem.

The maximum number of segments is $K_{\max} = $12. We use $M = $500 samples for the cross-validation procedure and the sampling probability is set to $f = $4/5. $B = $100 $N$ processes are sampled with each parameter configuration. For the Poisson-Gamma contrast, \changeAgain{}{as indicated in Section \ref{sec:ContrastPP},} we take the hyper-parameters $a=1$ and $b = 1/n$, where $n$ is the number of events in the process to be segmented.

\paragraph{Comparison with other approaches}
We found no available implementation for segmenting a continuous time point process. A first alternative and natural approach is to resort to a discrete time methodology, either by dividing the time interval into bins, or by considering the sequence of inter-event durations. As a consequence, we considered two alternative methods.
\begin{enumerate}
\item Discretization: we divided the time interval $(0, 1)$ into bins of width $\delta = (c n)^{-1}$ and associated the difference $N_{j \delta} - N_{(j-1)\delta}$ with each bin $j = 1, \dots, c n$. Each difference then follows a Poisson distribution. We considered four values for the discretization parameter $c$: $c \in\{0.5, 1, 2, 4\}$.
\item Inter-event times: we formed the sequence of inter-event times $(\Delta T_i = T_{i+1} - T_i)_i$, for $i = 1, \dots, n$, setting $T_0 = 0$. Each duration then follows an exponential distribution. 
\end{enumerate}
In both cases, we end up with a sequence of random variables, with respective parametric distributions (Poisson or exponential), whose parameters may contain discrete-time change-points that can be estimated.
Both models are implemented in the R package \texttt{changepoint} \changeNew{}{\citep{KiE14}}, which we used with the PELT algorithm and the mBIC penalty.
\changeNew{}{A comparison is given through Figure~\ref{fig:simOther-HausL2cum}.}


\paragraph{Marked Poisson process} 
The aim here is to see if the additional change-point information carried by the mark process helps the detection, and to know which of the two processes (the Poisson process and the mark process) allows better detection of the change-points. We use the same simulation design as previously for the change-point locations, fixing $\overline{\lambda}=100$. We consider two values for $\lambda_R$: $\lambda_R=1$ where the ground intensity of the Poisson process does not change from a segment to another, no signal in $\lambda$, and $\lambda_R=8$ where the change in the intensity is marked, signal in $\lambda$. For each value of $\lambda_R$, the parameter of the mark distribution $\rho(t)$ is either constant and equal to $0.1$, no signal in $\rho$, or alternates between $0.1$ and $0.005$ for the odd and the even segments respectively, signal in $\rho$. This leads to four scenarios in terms of detection according to the existence, or not, of changes in either $\lambda$ or $\rho$.

\paragraph{Piecewise Hawkes-type process}
We study the performance of the change-point detection in a piecewise Hawkes-type process, described in Section \ref{sec:HP}. 
We use the same segmentation $\taubf$ as in the Poisson case and transfer the difficulty of the change-point detection task to the values of $c_k$ involved in the intensity \eqref{eq:cond_int_php}. The parameter $c_k$ is set to $c_-$ in odd segments and to $c_+$ in even segments; we considered three values for the ratio $R = c_+ / c_-$:  $R \in \{2,3,6\}$. We fixed $\alpha=0.5$, $\beta=500$ and $c_+=500$ (thus, for $R = 2, 3, 6$, we have $c_- = 333, 250, 142.8$, respectively). By analogy with the stationary regime on each interval (under which the expected number of points is $c_k/(1-c_k \alpha/\beta)$), these parameters correspond to the pink curve displayed for the Poisson process on Figure \ref{fig:simPP-HausL2cum} ($\bar{\lambda}=562$).

We investigate the case where the parameters $(\alpha, \beta)$ are known (named "known" in the right panel of Figure \ref{fig:hawkes}) and the case where they are estimated using maximum likelihood on a homogeneous period with length $T_{\rm learn} \in \{1, 10\}$ and intensity $c_0 \lambda_0(t)$, where $c_0 = 100$.


\paragraph{Quality criteria} 
The performances are assessed according to the following criteria.
\begin{itemize}
  \item[-] The selected number of segments $\hat{K}$ (obtained with Algorithm \ref{algo:cv}). 
  \item[-] The Hausdorff distance $d(\taubf, \widehat{\taubf})$ between the true change-point locations $\taubf$ and the estimated ones $\widehat{\taubf}$ (with possibly different number of change-points, as in \cite{Chakar2017}). More specifically, defining
  \begin{align*}
  d_1(\taubf, \widehat{\taubf}) &= \max_k \min_\ell |\tau_k- \widehat{\tau}_\ell|, \qquad
  d_2(\taubf, \widehat{\taubf}) = d_1(\widehat{\taubf},\taubf), \\
  d(\taubf, \widehat{\taubf})  &= \max(d_1(\taubf, \widehat{\taubf}), d_2(\taubf, \widehat{\taubf})),
  \end{align*}
  $d_1$ indicates if each true change-point $\tau_k$ is close to an estimated one $\widehat{\tau}_\ell$ ($d_1$ will typically be small when $\widehat{K} \gg K$), whereas $d_2$ indicates if each estimated change-point $\widehat{\tau}_\ell$ is close to a true one $\tau_k$. A perfect segmentation results in both null $d_1$ and $d_2$ (and $d = 0$). 
  \item[-] The relative $L^2$-norm between the estimated and the true cumulative intensity functions. More specifically, denoting $\lambda(t)$ and $\widehat{\lambda}(t)$ the true and estimated intensity functions respectively and $\Lambda(t) =\int_0^t \lambda(u) \d u$ and  $\widehat{\Lambda}(t) =\int_0^t \widehat{\lambda}(u) \d u$ the corresponding cumulative intensity functions, we compute:
  $$
  \ell^2(\Lambda, \widehat{\Lambda}) = \left. \left( \int_0^1 (\widehat{\Lambda}(t) - \Lambda(t))^2 \d t \right) \right/ \lambdabar.
  $$
\end{itemize}

\subsection{Results}\label{sec:simuresults}

\subsubsection{Poisson process}
\paragraph{Model selection}
The left panel of Figure \ref{fig:simPP-HausL2cum} shows the mean number of the selected number of segments $\w{K}$ as a function of $\lambda_R$ (a measure of the difficulty of the task). The correct number of segments $K=6$ is better recovered when either the mean intensity $\lambdabar$ or the ratio $\lambda_R$ increases, as expected. More precisely, for a high mean intensity ($\lambdabar = 1000$), the correct number of segments $K = 6$ is recovered as soon as the ratio $\lambda_R$ reaches 3. In the typical case where $\lambdabar = 100$, the correct number of segments is obtained when $\lambda_R$ is about 10. When $\lambda_R = 1$, for all $\lambdabar$, the mean number of the selected number of segments $\widehat{K}$ is found to be 1, which is actually the correct number (as the intensity is constant in this case). When $\lambda_R$ is slightly greater than one, especially for small mean intensities $\lambdabar$, the model selection procedure tends to underestimate the number of segments. This behavior is classical and desired to avoid false detection \citep[see e.g.][]{cleynen2017model}.


\paragraph{Accuracy of the change-points}
The left panel of Figure \ref{fig:simPP-HausL2cum} represents the average of the Hausdorff distance as a function of the ratio $\lambda_R$. It shows the expected behavior in terms of accuracy of the locations of the estimated change-points. The Hausdorff distance $d(\taubf, \widehat{\taubf})$ decreases as either the mean intensity $\lambdabar$ or the ratio $\lambda_R$ increases. The distance is almost zero with a mean intensity $\lambdabar = 1000$ and a ratio $\lambda_R = 3$. In the typical case where $\lambdabar = 100$, the distance decreases very quickly as $\lambda_R$ increases, but remains above $.05$, meaning that some uncertainty remains about the precise locations of the change-point. Note that when $\lambda_R = 1$, the Hausdorff distance is almost 0 for all values of $\lambdabar$, simply because the selected number of segments is almost always 1, giving $\widehat{\taubf} = \{0, 1\}$, which is equal to the true $\taubf$.


\begin{figure}[hbtp]
  \begin{center}
    \includegraphics[width=.32\textwidth, trim=10 10 10 50, clip=]{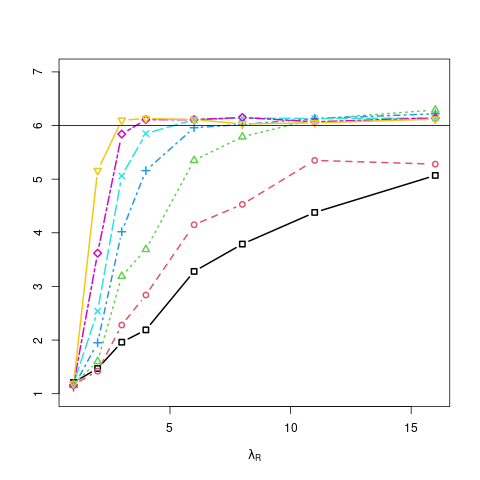}
    \includegraphics[width=.32\textwidth, trim=10 10 10 50, clip=]{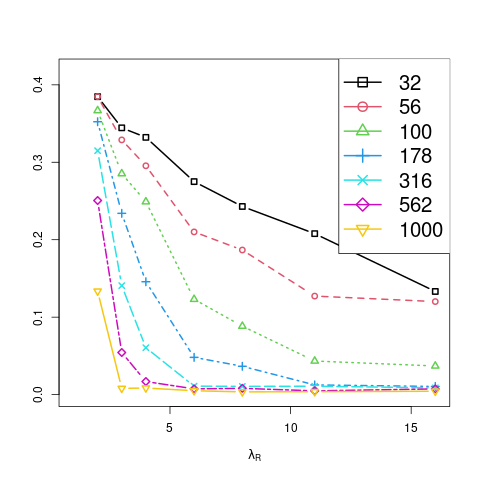}
    \includegraphics[width=.32\textwidth, trim=10 10 10 50, clip=]{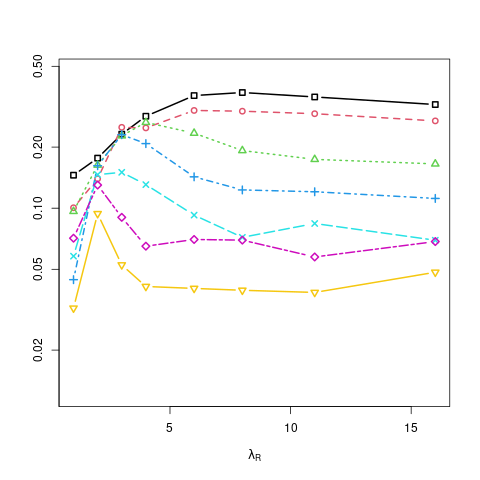}
    \caption{
    \changeNew{}{Left: Mean selected number of segments $\widehat{K}$ (averaged over $B=100$ replicates) as a function of the intensity ratio $\lambda_R$. 
    Center: Mean Hausdorff distance $d(\taubf, \widehat{\taubf})$ as a function of the ratio $\lambda_R$. 
    Legend panel = value of the mean intensity $\lambdabar$ (same for all panels). 
    Right: Mean relative distance $\ell_2(\Lambda, \widehat{\Lambda})$ between the true and estimated cumulated intensity function $\Lambda(t)$ as a function of the ratio $\lambda_R$.} 
    \label{fig:simPP-HausL2cum}
    }
  \end{center}
\end{figure}

The Hausdorff distance gives a synthetic measure of the proximity between the estimated change-points $\widehat{\taubf}$ and the true ones $\taubf$. To better illustrate the accuracy of the position of the estimated change-points, for a given configuration $(\lambdabar, \lambda_R)$, we gather all the detected change-points resulting from the $B=100$ simulations. Figure \ref{fig:simPP-HistTau} displays the distribution of the estimated locations for three configurations, chosen according to the results displayed in Figure \ref{fig:simPP-HausL2cum} (center): $\lambdabar=56, \lambda_R = 6$ (hard setting), $\lambdabar=100, \lambda_R = 8$ (intermediate setting), $\lambdabar=316, \lambda_R = 11$ (easy setting). The figure shows how the estimated $\widehat{\tau}_k$ concentrates around the true $\tau_k$ as the signal becomes more contrasted. \\
\change{}{When $\lambda_R=1$, the actual number of segments is equal to 1, so we do not show the Hausdorff distance in Figure 1 in this case. However, this configuration remains interesting in terms of the proportion of false alarms, i.e., when the estimated number of segments $\widehat{K}$ is greater than one. This proportion is fairly stable in relation to the average signal: for the seven values of $\lambdabar$ (ranging from 32 to 1000), we obtain false alarm proportions between $7\%$ and $16\%$ (mean = $13.1\%$, median = $14\%$), with no systematic trend.}\\
Note that the simulation procedure includes the selection of the number of segments $K$, so the different simulations do not provide the same number of estimated change-points $\widehat{\tau}_k$, which explains why the total number of estimated change-points varies from one configuration to another. 

\begin{figure}[hbtp]
  \begin{center}
    \begin{tabular}{ccc}
      $\displaystyle{\begin{array}{c} \lambdabar=56 \\\lambda_R = 6 \end{array}}$
      & &
      \hspace{-.05\textwidth}
      \begin{tabular}{c}
        \includegraphics[width=.7\textwidth, trim=25 0 50 50, clip=]{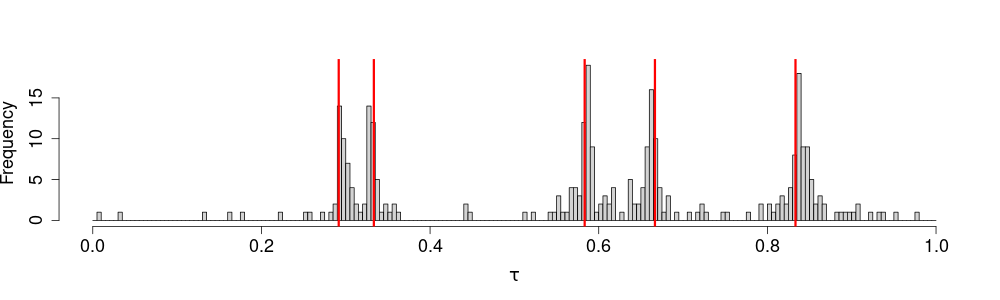} \\    
      \end{tabular}
      \\ \hline 
      $\displaystyle{\begin{array}{c} \lambdabar=100 \\\lambda_R = 8 \end{array}}$
      & &
      \hspace{-.05\textwidth}
      \begin{tabular}{c}
        \includegraphics[width=.7\textwidth, trim=25 0 50 50, clip=]{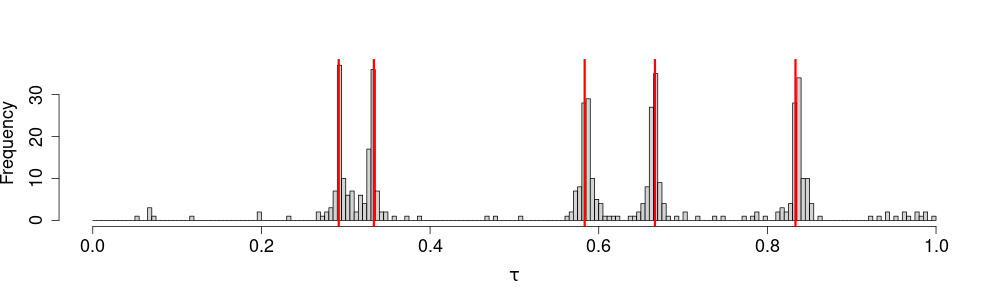} \\
      \end{tabular}
      \\ \hline 
      $\displaystyle{\begin{array}{c} \lambdabar=316 \\\lambda_R = 11 \end{array}}$
      & &
      \hspace{-.05\textwidth}
      \begin{tabular}{c}
        \includegraphics[width=.7\textwidth, trim=25 0 50 50, clip=]{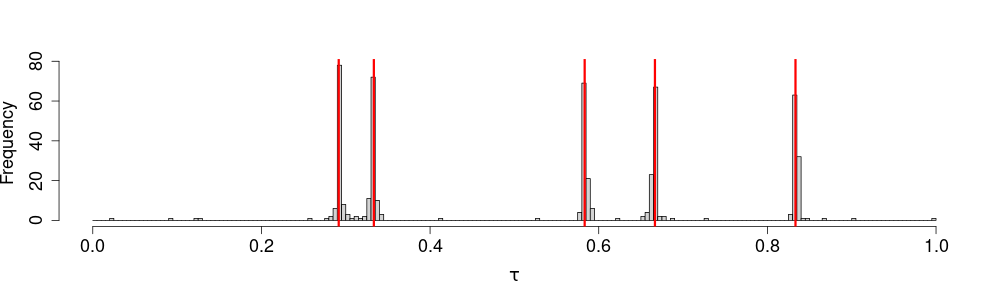} 
      \end{tabular}
    \end{tabular}
    \caption{Distribution of the estimated change-point locations $\widehat{\tau}_k$ across the $B = 100$ simulations for three different settings. 
    Top: $\lambdabar=56, \lambda_R = 6$, center: $\lambdabar=100, \lambda_R = 8$; bottom: $\lambdabar=316, \lambda_R = 11$. 
    Vertical red lines: true change-point locations $\taubf$.
    \label{fig:simPP-HistTau}}
  \end{center}
\end{figure}

\paragraph{Estimation of the intensity function}
The right panel of Figure \ref{fig:simPP-HausL2cum} represents the mean relative $\ell_2$ distance between the true and the estimated cumulated intensity function,  as a function of the ratio $\lambda_R$.  It shows that the estimation of the cumulated intensity function $\Lambda$ improves as the mean intensity $\lambdabar$ increases. More interestingly, the ratio $\lambda_R$ does not seem to strongly affect the accuracy of $\widehat{\Lambda}$. A possible explanation is that although a strong contrast between the intensity of neighboring segments gives better localization of the change-points, even a small error in the change-point location induces a high error in terms of $\lambda(t)$ or $\Lambda(t)$.

\paragraph{Comparison with alternative approaches}

\begin{figure}[hbtp]
    \begin{tabular}{m{.03\textwidth}m{.27\textwidth}m{.27\textwidth}m{.27\textwidth}}
        $\lambdabar$ & \multicolumn{1}{c}{$\widehat{K}$} & \multicolumn{1}{c}{Hausdorff} 
        & \multicolumn{1}{c}{Intensity} \\
        \hline
        $32$
        &
        \includegraphics[width=.3\textwidth, trim=10 10 10 50, clip=]{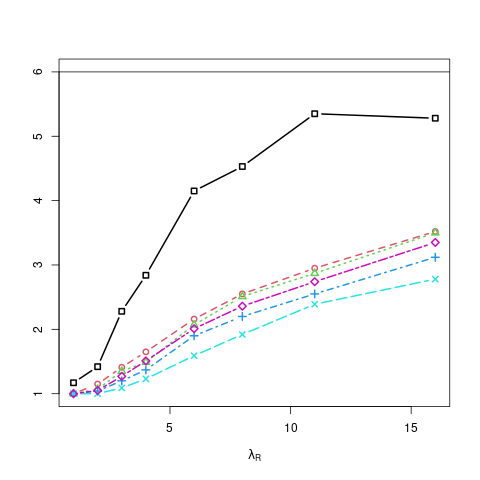}
        &
        \includegraphics[width=.3\textwidth, trim=10 10 10 50, clip=]{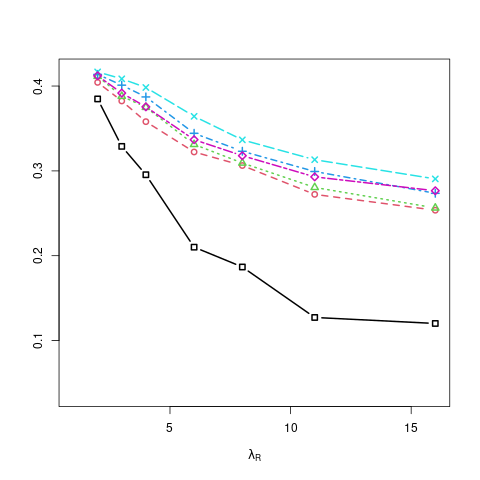}
        &
        \includegraphics[width=.3\textwidth, trim=10 10 10 50, clip=]{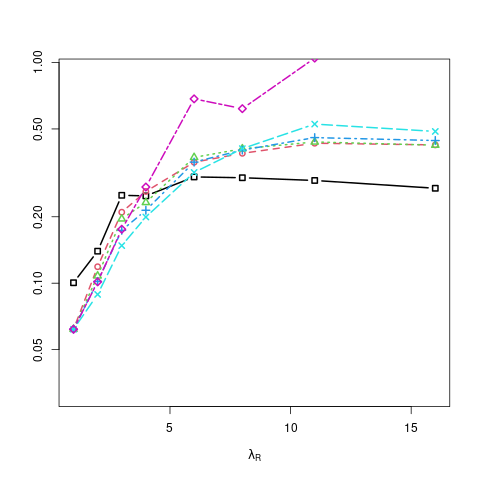} \\
        $178$
        &
        \includegraphics[width=.3\textwidth, trim=10 10 10 50, clip=]{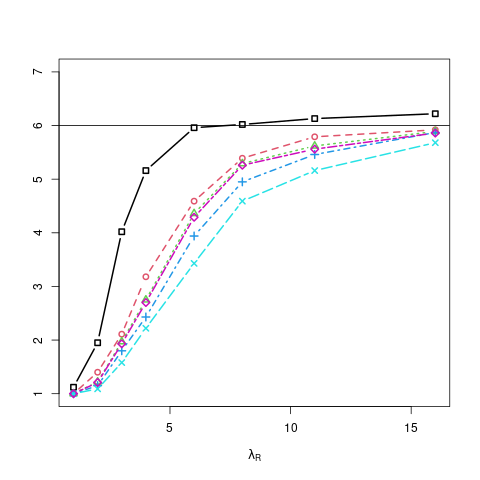}
        &
        \includegraphics[width=.3\textwidth, trim=10 10 10 50, clip=]{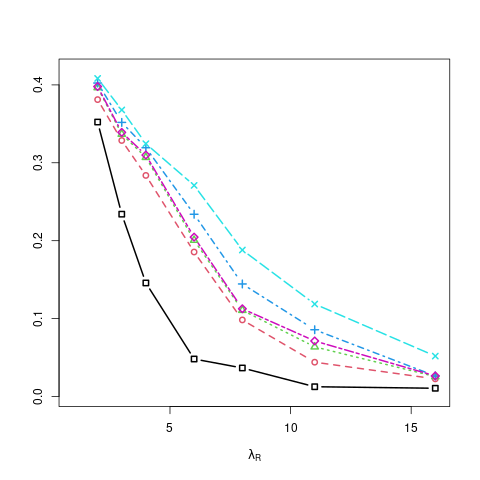}
        &
        \includegraphics[width=.3\textwidth, trim=10 10 10 50, clip=]{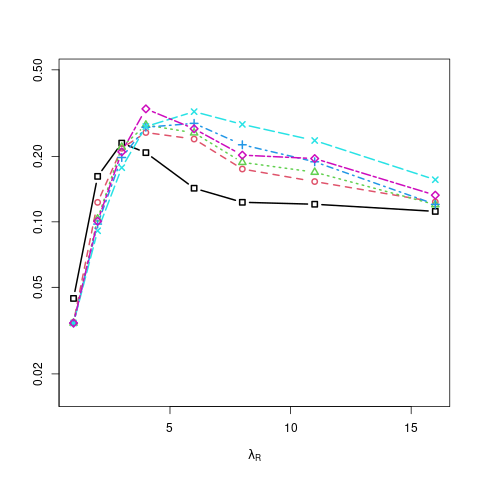} \\
        $562$
        &
        \includegraphics[width=.3\textwidth, trim=10 10 10 50, clip=]{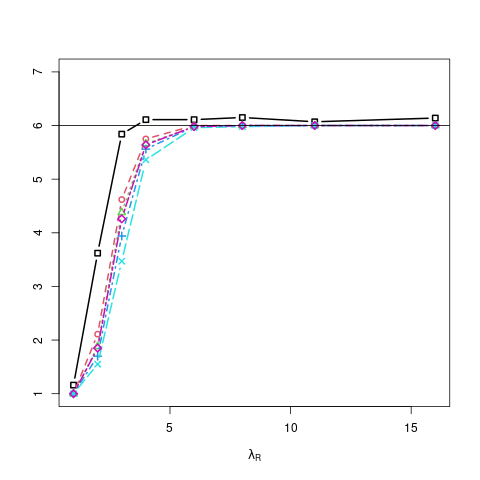}
        &
        \includegraphics[width=.3\textwidth, trim=10 10 10 50, clip=]{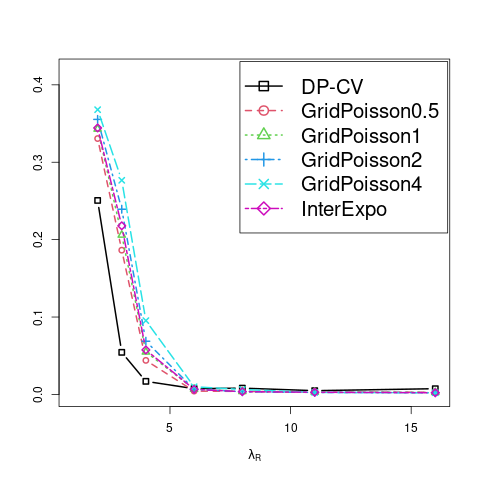}
        &
        \includegraphics[width=.3\textwidth, trim=10 10 10 50, clip=]{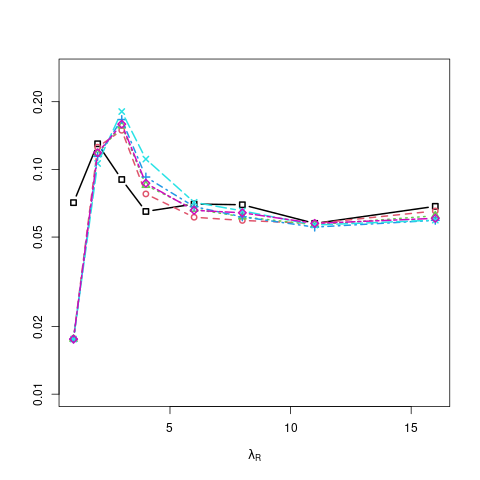}
    \end{tabular}
    \caption{
    \changeNew{}{Each line represents the same quality criteria as in Figure \ref{fig:simPP-HausL2cum} (from left to right: mean selected number of segments $\widehat{K}$, mean Hausdorff distance $d(\taubf, \widehat{\taubf})$, mean relative distance $\ell_2(\Lambda, \widehat{\Lambda})$, as functions of the intensity ratio $\lambda_R$)  for $\bar \lambda=52$ (first line), $\bar \lambda =178$ (second line) and $\bar \lambda=562$ (third line). 
    Legend panel (same for all): \texttt{DP-CV} \changeAgain{}{(for Dynamic Programming \& Cross-Validation)} is our method, GridPoisson is the change-point detection in Poisson distribution for $N$ observed on the four discrete grid and InterExpo is the change-point detection in an exponential distribution for $(\Delta T_i)_i$.}
    \label{fig:simOther-HausL2cum}
    }
\end{figure}

Figure \ref{fig:simOther-HausL2cum} shows a comparison between the proposed continuous time approach, \changeNew{}{called \texttt{DP-CV} (for Dynamic Programming \& Cross-Validation)}  and discrete time alternatives \changeNew{}{described at the end of Section~\ref{sec:simudesign}. We name \texttt{InterExpo} the PELT algorithm applied on the increments of the process, and \texttt{GridPoissonA} with $A \in \{0,1,2,4\}$ the method based on discretization of time with $c\in\{0.5, 1, 2, 4\}$.
}
We observe that when the number of events is low ($\bar{\lambda}=32$), all methods tend to underestimate the number of change-points, resulting in a poor Hausdorff distance. 
However, \texttt{DP-CV}, in black, offers the best results for all configurations. 
All methods improve as the average signal increases ($\bar{\lambda} = 178$ and $562$), but continuous time segmentation remains more accurate than discrete time segmentations. 
The results are much less contrasting in terms of cumulative intensity function estimation. 
\changeAgain{Interestingly to note that}{Interestingly,} the level of discretization (encoded in the parameter $c$) has a limited effect on the average performance of discrete Poisson approaches, and the exponential model gives results very similar to those of the discrete time Poisson model.

\subsubsection{Marked Poisson process}

\begin{table}[hbtp]
  \caption{Mean selected number of segments $\w{K}$ and mean Hausdorff distance $d(\taubf, \w{\taubf})$, from 200 repetitions. }
  \centering
  \begin{tabular}{ccccc}
  & \multicolumn{2}{@{}c@{}}{$\w{K}$} & \multicolumn{2}{@{}c@{}}{$d(\taubf, \w{\taubf})$}\\
  \hline
  $\lambda \backslash    \rho$ &  no signal & signal & no signal & signal \\
  \hline
    no signal    & 1.132 & 5.796 & 0.411 & 0.126 \\
      signal  & 5.411  & 5.998 &  0.112 & 0.056\\
  \end{tabular}
  \label{tab:MPP}
\end{table}

Table \ref{tab:MPP}  gives the average of the selected number of segments $\w{K}$ and the average of the Hausdorff distance $d(\taubf, \widehat{\taubf})$ for the four scenarios described in Section \ref{sec:simudesign}. When both the intensity and the mark parameters are constant, the average of $\w{K}$ is close to $1$, which is the true number, and thus the average of the Hausdorff distance is close to $0$. This shows that the procedure is not prone to overfit by detecting non-existing change-points. 

When both processes are significantly affected by the changes, the correct number of segments ($K=6$) is recovered and the estimated change-points are close to the true one (and the Hausdorff distance is close to $0$). 
Interestingly, in terms of power to detect change-points, there is no clear distinction between the cases where the signal is only in $\lambda$ or only in $\rho$, but the existence of a signal in both parts of the process increases this power.

\subsubsection{Hawkes-type process}
The left panel of Figure \ref{fig:hawkes} shows the estimated change-point in plain red lines and the true one in dotted red lines for a specific PHP with ratio $R=3$.

\begin{figure}[hbtp]
    \centering
    \includegraphics[width=0.35\textwidth]{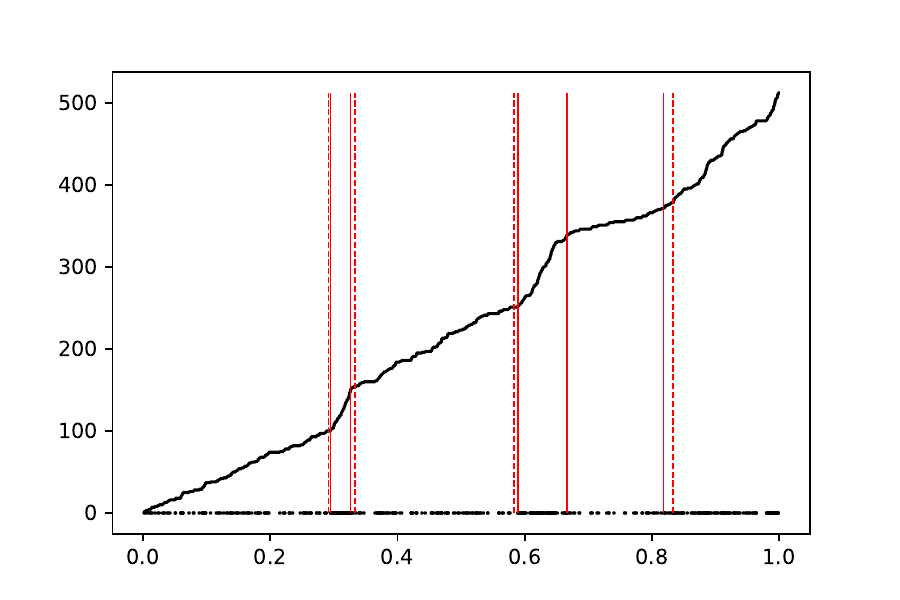}
    \qquad
    \includegraphics[width=.35\textwidth, height=.225\textwidth]{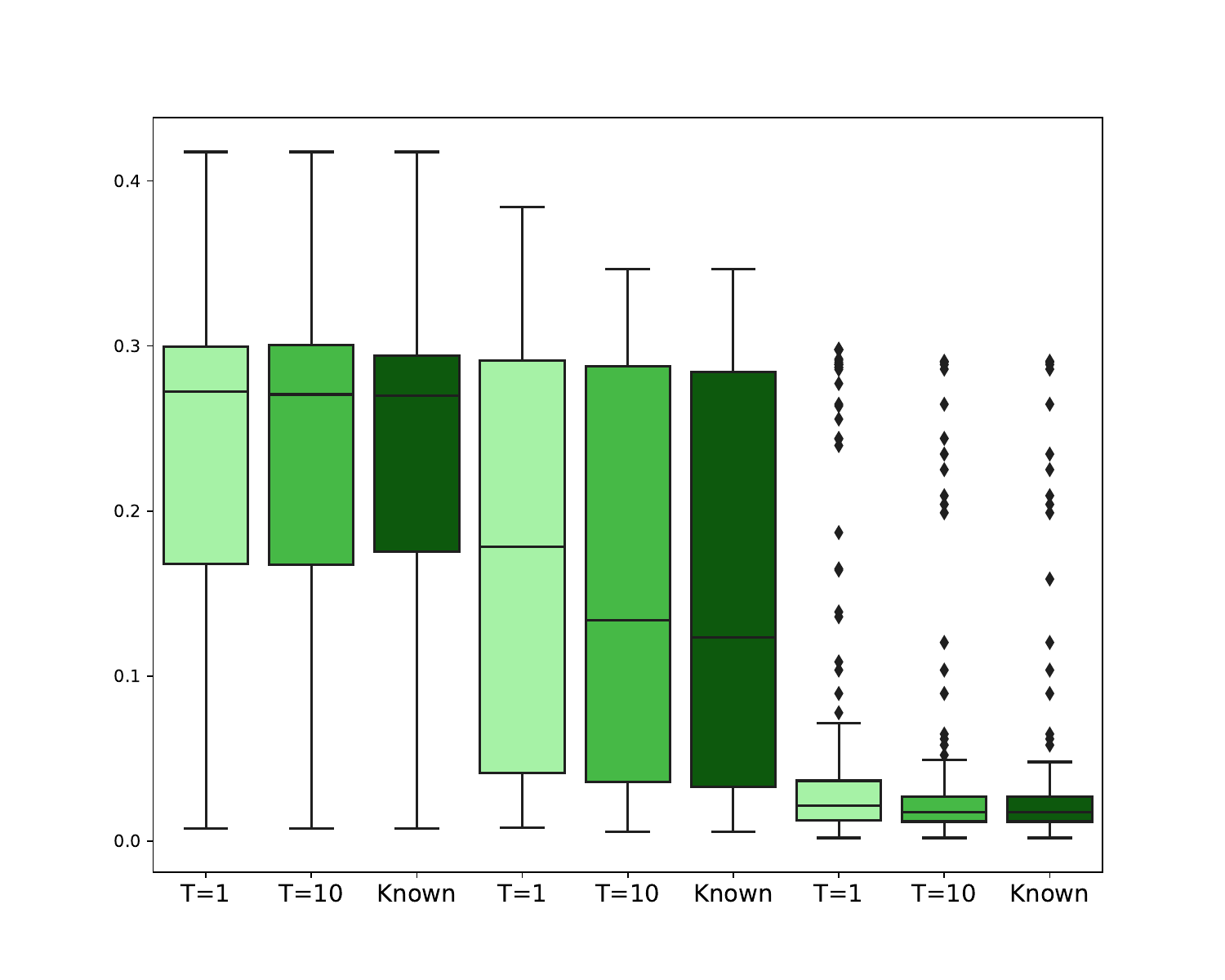}
    \caption{Left: Example with $R=3$. The counting process is the black line, with the events on the x-axis, the true change-points are the plain red lines and the estimated change-points are the dotted red lines. \\
    Right: Boxplot of the Hausdorff distance for 200 repetitions for $R \in \{2,3,6\}$ (three consecutive boxplots for each value of $R$)}
    \label{fig:hawkes}
\end{figure}

The summary of the results for the Hausdorff distance are presented in the right panel of Figure \ref{fig:hawkes}. 
The boxplots labeled "known" refer to the cases where the parameters $(\alpha, \beta)$ are known and given to the algorithm. We represent here three boxplots for each value of $R \in \{2,3,6\}$.
As expected, the lower is $R$, the harder is the task and, thus, the higher is the Hausdorff distance. We observe that the results when the coefficients are known are similar to those obtained for the pink curve on Figure \ref{fig:simPP-HausL2cum} for the Poisson process (which corresponds to the closest scenario in terms of parameters). 
Then, when the parameters have to be estimated on a period $[0, T_{\rm learn}]$, the mean Hausdorff distance for the estimated sequence of change-points is close to the one obtained when the parameters $(\alpha, \beta)$ are known, and even remarkably close for $T_{\rm learn} = 10$. This suggests that the proposed transformation of the process is actually efficient to detect change-points in a piecewise Hawkes-type process.

\section{Illustrations in volcanology and seismology} \label{sec:application}

In this section, for illustration, we analyze two datasets describing volcano activity. In Section \ref{sec:appPP}, only the eruption dates are taken into account, and we look for homogeneous segments in a Poisson process. In Section \ref{sec:appMPP}, the duration of each eruption is also considered and we look for change-points in a marked Poisson process.

\subsection{Poisson process} \label{sec:appPP}
We consider the eruptions of the Kilauea and Mauna Loa volcanoes in Hawaii, presented by \cite{HoB17}. Both datasets include recorded eruption dates from 1750 to 1983 for Kilauea Volcano and through 1984 for Mauna Loa Volcano. Over these periods, $n=63$ eruptions were observed for the first and $n=40$ for the second. The original data only reports the number of eruptions per year (ranging from 0 to 4). Continuous times were restored by associating each eruption with a uniformly distributed date, in the corresponding year. 

We used the cross-validation criterion defined in Section \ref{sec:selection} to determine the number of segments in each series of eruption dates. The criterion indicates the existence of four segments for the Kilauea and two segments (i.e. one change-point) for the Mauna Loa volcano. Note that this last criterion admits a local minimum at $K=4$ (see Figure \ref{fig:appPPselection} in \ref{app:appliPP}).

Figure \ref{fig:appPPsegmentation} gives the optimal segmentation of the Kilauea series in $K=4$ segments (three change-points located at the years 1918, 1934, 1952) and of the Mauna Loa series $K=2$ segments (one change-point located at the year 1843). 
We observe a better fit for the Kilauea series than for the Mauna Loa series. This may result from a conservative behavior of the model selection procedure when the number of events is small ($n = 63$ for the Kilauea volcano and $n=40$ for the Mauna Loa volcano). Figure \ref{fig:appPPappendix} in \ref{app:appliPP} shows the segmentations of the Mauna Loa series with $K = 3, 4, 5$ and $6$ segments.


\begin{figure}[hbtp]
  \centering
  \includegraphics[width=.35\textwidth, height=.3\textwidth, trim=0 0 20 30, clip=]{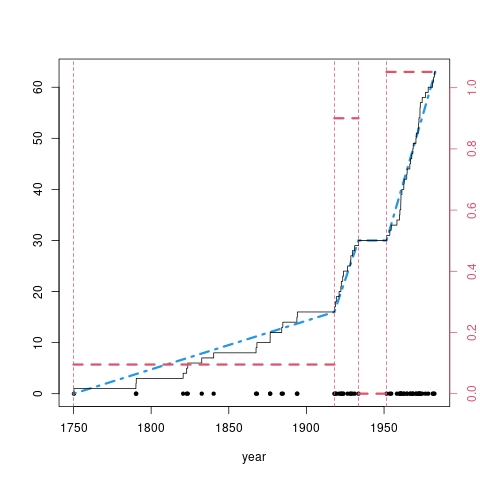}
  \qquad
  \includegraphics[width=.35\textwidth, height=.3\textwidth, trim=0 0 20 30, clip=]{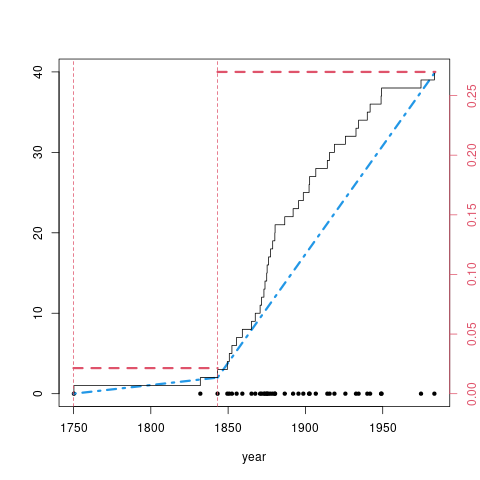}
  \caption{Eruptions of the Kilauea (left) and Mauna Loa (right) volcanos. Final segmentation: black bullets = eruptions times $T_i$, black solid line = observed count process $N(t)$, red vertical dotted line = estimated change-points $\widehat{\tau}_k$, blue dotted-dashed line = estimated cumulated intensity $\widehat{\Lambda}(t)$, red dashed horizontal lines = estimated piecewise constant intensity $\widehat{\lambda}(t)$ (referred to the right axis). The first vertical red dotted line is $\tau_0=0$. \label{fig:appPPsegmentation}}
\end{figure}

\subsection{Marked Poisson process}  \label{sec:appMPP}
The data comes from the technical report \citep{etna}. It consists of Mount Etna volcano \changeAgain{flunk}{flank} eruption data between 1607 and 2008. Indeed,  flunk eruption constitutes one of the most important threats for an assessment of volcanic risk according to the authors. The first variable is the date of the eruption of the volcano located in the North Eastern part of the Sicily Island, and the second one is the volume of lava spread. There are $n=63$ events.

We study these data with the inhomogeneous Poisson model then with the marked Poisson model that we presented previously. We search for the best number of segments $K$ in the collection $\{1, \ldots, 10\}$.
The chosen dimension is $\w{K}=2$ with both models (see Figure \ref{fig:etna-contrasts} in \ref{app:appliPP}) and the obtained segmentations are given in Figure \ref{fig:etnaMPP}. We see that the segmentations are different and we believe that we can see the influence of the mark on the second graph. Indeed, on the one hand the Poisson model without marks chooses a change-point at 1968 because after that (as it is explained for example in \cite{report}) the eruptions occur very closely in time until 2001, which justifies a change of regime. On the other hand, the marked process chooses a change-point at 1755, probably mainly because the marks are larger before this time than after (indeed, the report \citep{report} explains that probably the lava measured after 1755 was mostly buried under the products of additional activity). Furthermore, in Figure \ref{fig:etnaMPP-K3} where we represent the segmentation obtained with the marked Poisson model,  imposing $K=3$, we can see that the method gives a segmentation with two change-points which are close to the single change-point found by the Poisson process model and the single change-point found by the marked Poisson process model. 


\begin{figure}[hbtp]
\centering
\includegraphics[width=.35\textwidth, trim=0 0 20 30, clip=]{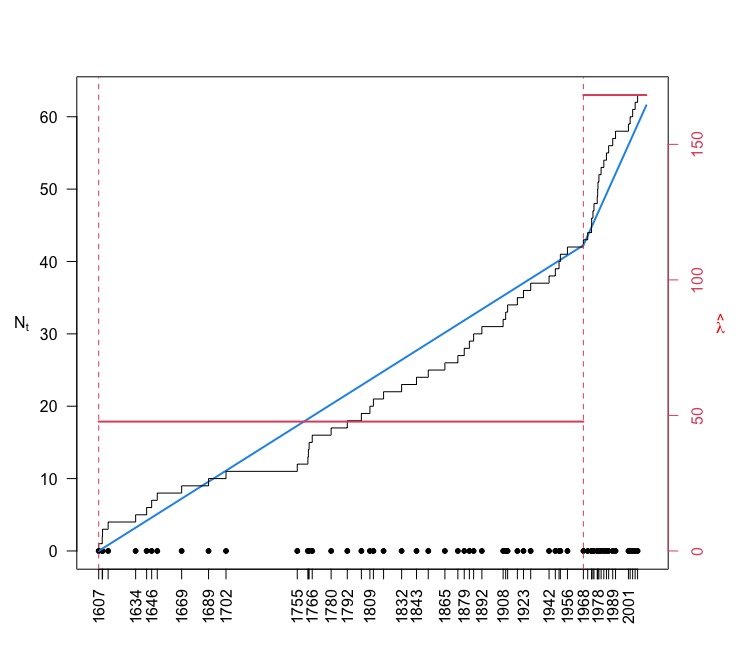}\\
\includegraphics[width=.35\textwidth, trim=0 0 20 30, clip=]{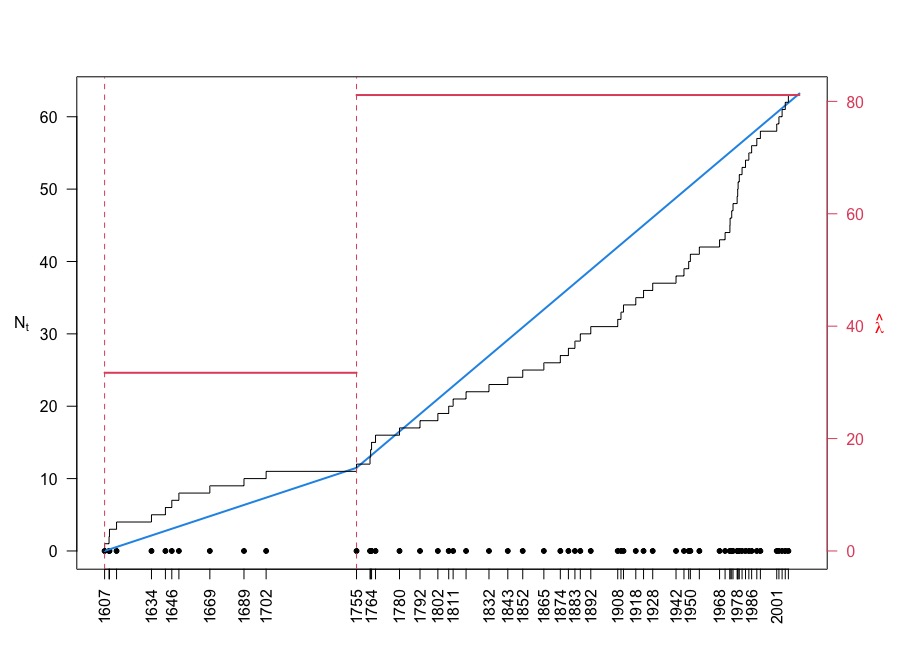}
\qquad
\includegraphics[width=.35\textwidth, trim=0 0 20 30, clip=]{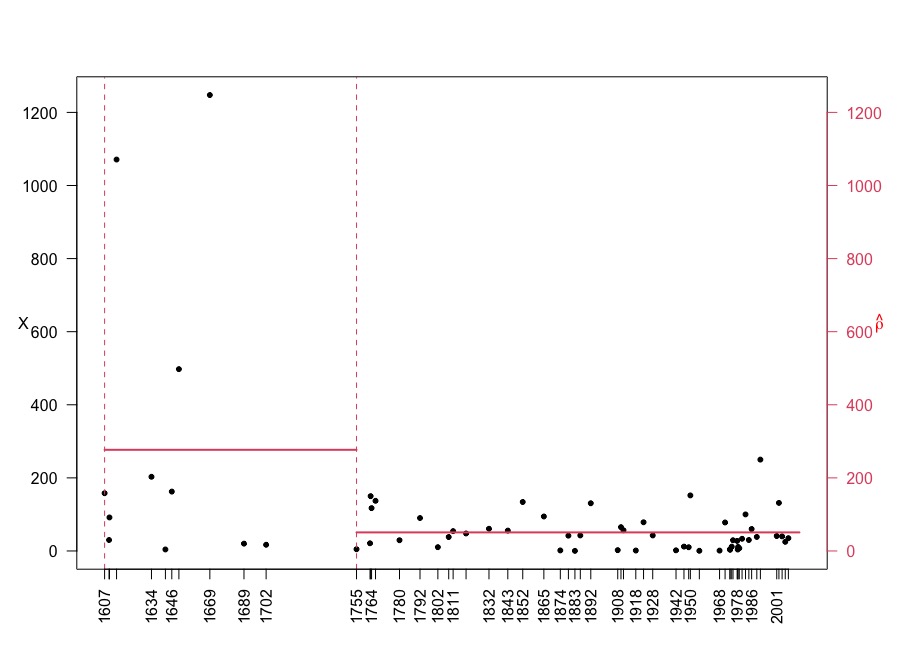}
\caption{Eruptions of Etna.  Top line: Poisson model, bottom line: marked Poisson model.  
(Left graphs) Final segmentations: black bullets = eruptions times $T_i$, black solid line = observed count process $(N_t)$, red vertical dotted lines = estimated change-points $\widehat{\tau}_k$, blue dotted-dashed line = estimated cumulated intensity $\widehat{\Lambda}(t)$, red dashed horizontal lines = estimated piecewise constant intensity $\widehat{\lambda}(t)$ (referred to the right axis). Bottom right graph: black bullets = marks $X_i$,  red vertical dotted lines = estimated change-points,  red dashed horizontal lines = estimated piecewise constant parameter of the exponential distribution of the marks  $\widehat{\rho}$. }
\label{fig:etnaMPP}
\end{figure}
 
\begin{figure}[hbtp]
  \centering
  \includegraphics[width=.35\textwidth, trim=0 0 20 30, clip=]{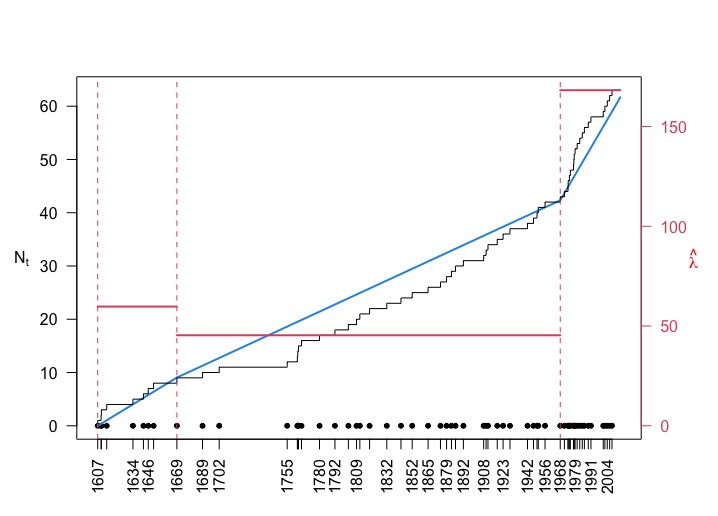}
  \qquad
  \includegraphics[width=.35\textwidth, trim=0 0 20 30, clip=]{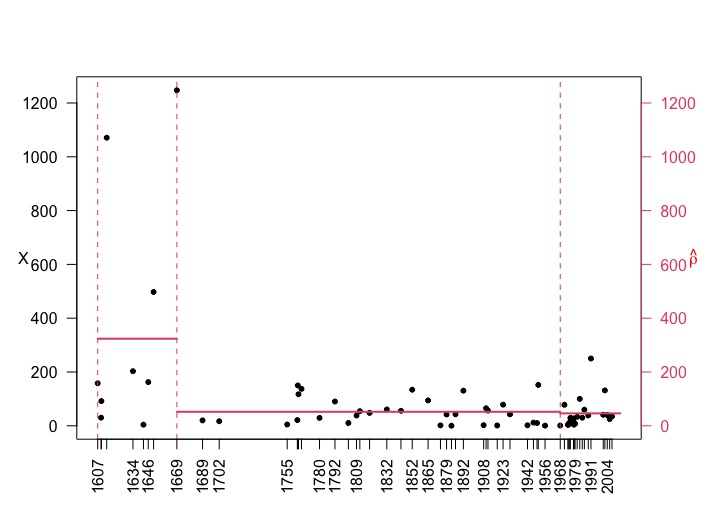}
  \caption{Eruptions of Etna.  Marked Poisson model.  Same legend as the top line of Figure \ref{fig:etnaMPP} with ${K}=3$.}
  \label{fig:etnaMPP-K3}
\end{figure}

\subsection{Self-exciting process}  \label{sec:appHP}


We now examine the sequence of earthquakes and aftershocks in Thailand, described in \cite{ptprocess}, and available in the R-package \texttt{PtProcess}. The complete dataset includes the magnitudes of the shocks, which we use here solely to characterize the main events to be detected. The three significant events are the earthquakes that occurred in Phuket in December, 25, 2004,  and in Sumatra in March, 28, 2005, and September, 12, 2007, corresponding to the largest magnitudes (around 8.5) as shown in the right panel of Figure \ref{fig:phuket}. 
We apply the PHP model to detect these main events using only the recorded occurrences, excluding the magnitudes. Since the parameters of the underlying process, $\theta$, are unknown, we first fix the parameter $\beta$ and then run the algorithm over a grid of $\alpha$ values, selecting the one that maximizes the log-likelihood. 

Initially, we analyzed the observed process within a Poisson framework (as outlined in Sections \ref{sec:ModelPP}, \ref{sec:GeneralCP} and \ref{sec:ContrastPP}), but were unable to identify the three main events, even with a larger number of change-points. 
We then shifted to the PHP model, estimating the parameters $\alpha$ and $\beta$ using the second method described at the end of Section \ref{sec:HP}, with $K=5$ segments. We observed a minimal influence of the $\beta$ value on the change-point detection in this example, so we fixed it at $8$.
The results shown in Figure \ref{fig:phuket} indicate that we accurately detected the three main events (with 15 days of delay for the second one. We detected an additional event on September, 29, 2007, which seems to end a very high seismic activity.
\begin{figure}[hbtp]
  \centering
  \includegraphics[width=.35\textwidth, trim=0 0 20 30, clip=]{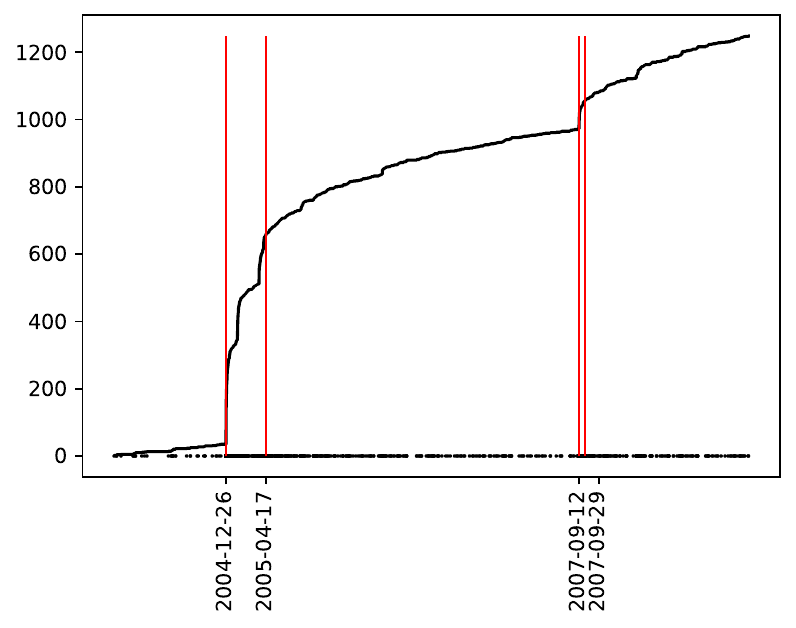}
  \qquad
  \includegraphics[width=.35\textwidth, trim=0 0 20 30, clip=]{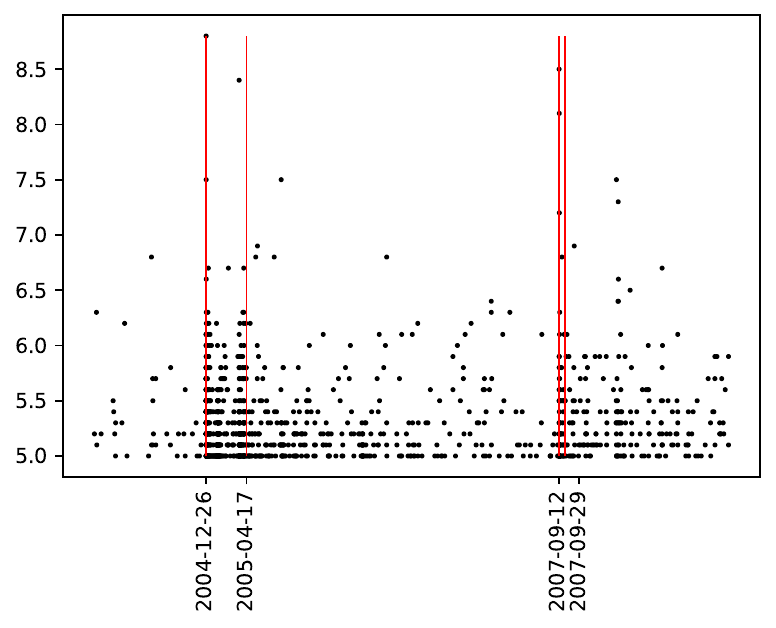}
  \caption{Aftershock sequence of the Phuket earthquake. Counting process (black line) and events (dots), together with estimated change-points in red lines and the dates on the $x$-axis. Right graph: magnitudes (not used for the estimation).}
  \label{fig:phuket}
\end{figure}



\section{Discussion}\label{sec:discussion}

We present a comprehensive frequentist framework for detecting multiple change-points in the intensity of a Poisson process. As is common in change-point detection, the inference procedure consists of two steps: (1) segmenting the observed path $N$ into a predetermined number of segments, and (2) selecting this number. For step (1), we demonstrate that any contrast function satisfying both segment-additivity and \changeAgain{}{strict} concavity allows for the optimal segmentation into $K$ segments to be achieved exactly and efficiently using a dynamic programming algorithm.
Note that the use of dynamic programming fundamentally relies on the segment-additivity property and that traditional methods like maximum likelihood or least-squares inference fit within this framework. For step (2), we propose employing a cross-validation strategy.

\changeAgain{}{The practical application of our method is subject to certain scope limitations. First, the method's performance is inherently sensitive to low event counts. As demonstrated in the simulation study, in scenarios with a very scarce signal (e.g., a low mean intensity $\bar{\lambda}$), the CV procedure exhibits a conservative behavior. Consequently, the localization accuracy diminishes when the available information is severely limited. Second, the extension to self-exciting PHP relies fundamentally on the time-scaling theorem, which requires the baseline increase and decay parameters ($\alpha$ and $\beta$) to be known or well-estimated. If a homogeneous learning period is not available to pre-estimate these parameters, the procedure necessitates a grid search over the parameter space (as performed in our seismology illustration). This fallback strategy increases the computational burden and can be sensitive to the chosen grid resolution. Furthermore, the transformation to a Poisson process strictly assumes that regime switches act as a multiplicative factor $c_k$ on the baseline intensity; structural deviations from this proportional model fall outside the current theoretical scope of the transformation.}

Besides, most model selection criteria for determining the number of segments $K$ involve penalized contrasts, where the penalty term is dependent on $K$ \citep[as in the classical BIC and AIC, or in][]{Leb05,lavielle2005using,jackson2005,killick2012} or both on the number of segments and their lengths \cite{ZhS07}.
As noted by \cite{jackson2005}, if the penalty term is proportional to $K$, it can be spread in the cost of each segment, allowing for the integration of both tasks into a single dynamic programming step, thus achieving segmentation and model selection simultaneously.
Our framework can also include more general penalized contrasts, provided that the penalty also satisfies the same admissibility condition.

The computational cost of dynamic programming is quadratic in the number of events $n$ and linear in the number of segments $K$, that is $\cO(n^2 K)$, which can be too demanding when dealing with a large number of events.
Time efficiency can be enhanced by utilizing the pruned version of dynamic programming proposed by \cite{killick2012}, which yields a complexity almost linear in $n$.

Finally, in many scenarios, multiple simultaneous trajectories are observed. In such cases, detecting changes that affect all series concurrently becomes the goal. If the associated contrast remains concave with respect to segment lengths, this problem can be addressed with a complexity comparable to that of the single trajectory case.

\section{Proofs} \label{app:proofs}

\subsection{Proofs of Proposition \ref{prop:Concavity_C}} \label{sec:proof_concavity}
Let fix $\nubf \in \Upsilon^{K,n}_\star$ and  take $\taubf \in \cM_\nu^K (N)$. 
Denoting by $C'$ and $C''$ the first and second derivative of $C$ with respect to its second argument, we have, for all $1 \leq k < K$,
\begin{align*}
\frac{\partial \gamma(\taubf) }{\partial \tau_k} & = C^\prime (\nu_k,\Delta \tau_k) - C^\prime (\nu_{k+1},\Delta \tau_{k+1}),  & 
\frac{\partial^2 \gamma(\taubf) }{\partial \tau_k \partial\tau_{k+1}} & = -  C^{\prime \prime}  (\nu_{k+1},\Delta \tau_{k+1}) \quad (\text{if $k < K-1$}),\\
\frac{\partial^2 \gamma(\taubf) }{\partial \tau_k^2} & = C^{\prime \prime}  (\nu_k,\Delta \tau_k) + C^{\prime \prime}  (\nu_{k+1},\Delta \tau_{k+1}).
\end{align*}
The Hessian matrix of the contrast $\gamma$ is the following $(K-1)\times (K-1)$ matrix
\begin{equation*}\label{eq:hess}
H_K = \left(
\begin{matrix}
-(A_1+A_2) & A_2 & 0 &\ldots \\
A_2 & -(A_2+A_3) & A_3 & \ldots \\
0 & \ddots   &  \ddots&  \ddots\\
 \vdots &0 & A_{K-1} & -(A_{K-1}+A_{K})
\end{matrix}\right),
\end{equation*}
where $A_k: = -C^{\prime \prime}  (\nu_k,\Delta \tau_k)$.
For all vector $u \in \R^{K-1}\setminus\{0\}$, we have that
\begin{align*}
  u^t H_K u 
  & =  - \sum_{k = 1}^{K-1} u_k^2(A_k+A_{k+1})+2 \sum_{k = 1}^{K-2} u_kA_{k+1}u_{k+1},\\
  & =  - \sum_{k = 1}^{K-2} (u_k-u_{k+1})^2 A_{k+1}-u_1^2 A_1-u_{K-1}^2 A_K.
\end{align*}
Using the concavity assumption \ref{hyp:concavity} of the contrast function on each segment we have that $A_k \geq 0$ for all $k$, thus $u^t H_K u \leq 0$, which concludes the proof \change{}{ in the concave case}. \change{}{Under the strict version of assumption \ref{hyp:concavity}, all $A_k$ are positive, so $u^t H_K u < 0$, which concludes the proof for the strictly concave case.} 

\subsection{Proof of Proposition \ref{prop:link_compensators}}
If $t\in [\tau_k, \tau_{k+1}[$, 
$$\Lambda(t)= \int_0^t \lambda(u)du = \int_0^t \sum_{\ell =1}^{k+1} c_\ell \one_{[\tau_{\ell-1}, \tau_\ell[}(u) \lambda_0(u) du$$
then 
\begin{align*}
  \Lambda(t)
  = & \sum_{\ell =1}^{k} c_\ell \int_{\tau_{\ell-1}}^{\tau_\ell}  \lambda_0(u) du
+ c_{k+1} \int_{\tau_k}^t \lambda_0(u) du\\
  & = \sum_{\ell =1}^{k} c_\ell (\Lambda_0(\tau_\ell)- \Lambda_0(\tau_{\ell-1}) + c_{k+1}(\Lambda_0(t)- \Lambda_0(\tau_k)).
\end{align*}

\subsection{Proof of Theorem \ref{th:Modified_time_rescaling}}
Denoting by ${T_1, \dots, T_n}$ the event times of the Hawkes-type process $N$ defined in Section \ref{sec:HP}, and by $(H_t)_t$ the history of the process, the likelihood of $N$ is
\begin{align*}
f_N(T_1, \dots T_n) 
:= p_H(N;\taubf, \cbf, \alpha,\beta ) 
& = \left(\prod_{i=1}^n \lambda(T_i)\right) e^{-\int_0^1 \lambda(s) \d s}.
\end{align*}
We now split the interval $[0, \tau_K)$ into the intervals $I_k = [\tau_{k-1}, \tau_k)$ for $k = 1, \dots K$. Denoting by $n_k$ the number of event times that occur in interval $I_k$, because $\lambda(t) = c_k \lambda_0(t)$ within segment $I_k$, we get
\begin{align} \label{eq:joint_density_events1}
f_N(T_1, \dots T_n) 
    & =  \prod_{k=1}^K
   \left(\prod_{T_i \in I_k} \lambda(T_i)\right) e^{-\int_{\tau_{k-1}}^{\tau_{k}} \lambda(s) \d s} \nonumber \\
    & = 
    \left(\prod_{i=1}^n \lambda_0(T_i) \right)
    \left(\prod_{k=1}^K c_k^{n_k} e^{-c_k (\Lambda_0(\tau_k) - \Lambda_0(\tau_{k-1}))}\right), 
\end{align}
which gives Equation \eqref{eq:transformed_poisson_loglik}.\\

For event times $(t_i)_i$, we define the transformed inter-event times $u_i :=\Lambda_0(t_i) - \Lambda_0(t_{i-1})$, for $i \in \{1, \dots, n\}$ (with $t_0 = 0$), $\psi$ as the function defined over the cube $[0, \tau_K)^n$ by
$$
\psi : (t_1, \dots , t_n) \mapsto (u_1, \dots, u_n)
$$
and $J_\psi$ as its Jacobian matrix.
Using a multivariate change of variable theorem, the joint density function $f_{\mathbf{u}}$ of $\{u_1, \dots, u_n\}$ is 
\begin{align} \label{eq:main1_joint_density}
    f_{\mathbf{u}}(u_1, \dots, u_n) = 
    f_N\left(\psi^{-1} (u_1, \dots, u_n)\right)
    / 
    \left|J_{\psi} \left(\psi^{-1} (u_1, \dots, u_n)\right)\right|,
\end{align}
This change of variable is valid because the function $\psi$ is one-to-one, and the Jacobian determinant $\left|J_{\psi} (t_{1}, \dots , t_{n_k})\right|$ does not vanish, allowing us to obtain \eqref{eq:main1_joint_density}.
The first assertion is straightforward as $\Lambda_0$ is a bijective function.
The second assertion can be confirmed by observing that $J_{\psi} (t_1, \dots , t_n)$ is a lower triangular matrix, leading to its determinant being the product of its diagonal elements, that is
\begin{equation} \label{eq:determinant_jacobian_mcv1}
    \left|J_{\psi} (t_1, \dots , t_n)\right| 
    = \prod_{k=1}^K \prod_{t_i \in I_k} \frac{\lambda(t_i)}{c_k}
    = \prod_{i=1}^{n} \lambda_0(t_i), 
\end{equation}
which doesn't vanish since $\lambda(t)>0$ and $c_k>0$ by assumption.
Inserting \eqref{eq:joint_density_events1} and \eqref{eq:determinant_jacobian_mcv1} in expression \eqref{eq:main1_joint_density}, we get
\begin{align*}
    f_{\mathbf{u}}(u_1, \dots, u_n)  
    = 
    \prod_{k=1}^K c_k^{n_k}
    e^{-c_k(\Lambda_0(\tau_k) - \Lambda_0(\tau_{k-1}))} 
\end{align*}
which is the joint distribution of the inter-event times of a heterogeneous Poisson process with piecewise constant intensity \eqref{eq:cond_int_php}.

\section*{Acknowledgments}
This work has been conducted within the FP2M federation (CNRS FR 2036). This work is also part of the 2022 DAE 103 EMERGENCE(S) - PROCECO project supported by Ville de Paris. We are grateful to the INRAE MIGALE bioinformatics facility (MIGALE, INRAE, 2020. Migale bioinformatics Facility, doi: 10.15454/1.5572390655343293E12) for providing computing resources.



\newpage
\appendix

\section{Contrasts} \label{app:contrasts}
In this section, we give some details about the calculations of the contrasts for both the Poisson and the marked Poisson models given in Sections \ref{sec:ContrastPP} and \ref{sec:MPP} respectively.

\paragraph{Least-squares Poisson contrast}
The least-squares criterion of a given path $N$ and for a fixed $\nubf \in \Upsilon^{K,n}_\star$ and any $\taubf \in \mathcal{P}_\nubf^K(N)$ is
\begin{equation} \label{eq:LS}
  LS_P  (N ; \taubf, \lambdabf)=\sum_{k=1}^K \left ( \Delta \tau_k  \ \lambda_k^2- \Delta N_k \lambda_k \right),
\end{equation}
that is minimal for the same intensities $\widehat{\lambda}_k(\taubf) = \frac{\Delta N_k}{\Delta \tau_k}$, $1 \leq k \leq K$. The associated contrast function of $\taubf$ is therefore as follows
\begin{equation} \label{contrast:LS}
 \tilde{\gamma}_P (\taubf)
 = LS_P  (N ; \taubf, \w{\lambdabf}) =- \sum_{k=1}^K  \frac{\nu_k}{\Delta \tau_k}= \sum_{k=1}^K  C(\nu_k,\Delta \tau_k).
\end{equation}
Similar to the likelihood contrast, it is straightforward to see that both assumptions \ref{hyp:additivity} and \ref{hyp:concavity} are satisfied. Moreover, it shares the same undesirable property: for $K>2$, the optimal segmentation inevitably contains segments of null length.  

\paragraph{Poisson-Gamma likelihood-based contrast}
Recall that the intensities $\lambda_k$ are assumed to be independent random variables and follow a gamma distribution with parameters $a>0$ and $b>0$, denoted $\Gam(a,b)$. The distribution of $\lambdabf$ is thus
\begin{equation*}
  p_{G}(\lambdabf ; a, b)=\prod_{k=1}^K   \frac{b^a}{\Gamma(a)} \ \lambda_k^{(a-1)} \ e^{-b \lambda_k},
\end{equation*}
where $\Gamma$ is the gamma function: $\Gamma(a) = \int_0^{+\infty} e^{-t} t^{a-1} \d t$. 
For a given $\taubf$, the joint distribution of $(N, \lambdabf)$ is therefore
\begin{equation*}
  p_{PG}(N, \lambdabf ; \taubf, a, b)  =  p_{P}(N ;  \lambdabf, \tau ) \  p_{G}(\lambdabf ;  a, b) = \prod_{k=1}^K \frac{b^a}{\Gamma(a)}   \lambda_k^{(\Delta N_k+a-1)} \ e^{-(\Delta \tau_k+b) \lambda_k},
  \end{equation*}
and the marginal distribution of $N$ is
\begin{equation*}
  p_{PG}(N ; \taubf, a, b)= \int  p_{PG}(N, \lambdabf ; \taubf, a, b)   \d \lambdabf= \prod_{k=1}^K \dfrac{b^a \; \Gamma(\Delta N_k+a)}{\Gamma(a) \; (\Delta \tau_k+b)^{\Delta N_k+a}}.
\end{equation*}
For a fixed $\nubf \in \Upsilon^{K,n}_\star$ and all $\taubf \in \mathcal{P}_\nubf^K(N)$, the proposed contrast, the so-called Poisson-Gamma contrast, is then
\begin{equation*} 
  \gamma_{PG} (\taubf) = - \log{p_{PG}(N \mid \taubf; a, b)}=\sum_{k=1}^K \left (-a \log{b} +\log{\Gamma(a)}  +\widetilde{a}_k \log{\widetilde{b}_k} - \log{\Gamma(\widetilde{a}_k)} \right),
 \end{equation*}
where $\widetilde{a}_k =  \nu_k+a$, $\widetilde{b}_k =  \Delta \tau_k+b$. 

\paragraph{Marked Poisson-Gamma-Exponential-Gamma likelihood-based contrast} Recall that the $\lambda_k$'s and $\rho_k$'s are assumed to be independent random variables and follow a Gamma distribution with parameters $a_\lambda>0$ and $b_\lambda>0$ and a Gamma distribution with parameters $a_\rho>0$ and $b_\rho>0$, respectively.  For a given $\taubf$, the joint distribution of $(N, X, \lambdabf,\rhobf)$ is thus
\begin{align*}
  p_{MPGEG}(N, X, \lambdabf,\rhobf ; \taubf, a_\lambda, b_\lambda, a_\rho, b_\rho)  
  & =   p_{P}(N, X ; \taubf, \lambdabf,\rhobf ) \  p_{G}(\lambdabf ;  a_\lambda, b_\lambda) \  p_{G}(\rhobf ;  a_\rho, b_\rho),  \\
  & =  \prod_{k=1}^K \frac{{b_\lambda}^{a_\lambda}}{\Gamma(a_\lambda)}   \lambda_k^{(\Delta N_k+a_\lambda-1)} e^{-(\Delta \tau_k+b_\lambda) \lambda_k} \\
  & \quad \times \prod_{k=1}^K \frac{{b_\rho}^{a_\rho}}{\Gamma(a_\rho)}   \rho_k^{(\Delta N_k+a_\rho-1)} \ e^{-(S_k+b_\rho) \rho_k},
\end{align*}
and the marginal distribution of $(N, X)$ is
\begin{align*}
  p_{MPGEG}(N, X ; \taubf ,  a_\lambda, b_\lambda, a_\rho, b_\rho)
  & = \iint  p_{PGEG}(N, \lambdabf,\rhobf ; \taubf , a_\lambda, b_\lambda, a_\rho, b_\rho) \  \d \lambdabf \  \d \rhobf, \\
  & = \prod_{k=1}^K \dfrac{{b_\lambda}^{a_\lambda} \; \Gamma(\Delta N_k+a_\lambda)}{\Gamma(a_\lambda) \; (\Delta \tau_k+b_\lambda)^{\Delta N_k+a_\lambda}} \ \dfrac{{b_\rho}^{a_\rho} \; \Gamma(\Delta N_k+a_\rho)}{\Gamma(a_\rho) \; (S_k+b_\rho)^{\Delta N_k+a_\rho}}.
\end{align*}
For a fixed $\nubf \in \Upsilon^{K,n}_\star$ and for any $\taubf \in \mathcal{P}_\nubf^K(N)$, the MPGEG contrast is then
\begin{align*} 
  \gamma_{MPGEG} (\taubf)
  & = -  \log{p_{MPGEG}(N ; \taubf , a_\lambda, b_\lambda, a_\rho, b_\rho)},  \\
  & =  \sum_{k=1}^K \left ( \widetilde{a}_{k, \lambda} \log{\widetilde{b}_{k, \lambda}} - \log{\Gamma(\widetilde{a}_{k, \lambda})} \right) +\left (  \widetilde{a}_{k, \rho} \log{\widetilde{b}_{k, \rho}} - \log{\Gamma(\widetilde{a}_{k, \rho})} \right) \nonumber\\
  & \quad + K \left ( -a_\lambda \log{b_\lambda } +\log{\Gamma(a_\lambda )}  -a_\rho \log{b_\rho } +\log{\Gamma(a_\rho )} \right ), \nonumber
\end{align*}
where $\widetilde{a}_{k, \lambda} =  \nu_k+a_\lambda$, $\widetilde{b}_{k, \lambda} =  \Delta \tau_k+b_\lambda$, $\widetilde{a}_{k, \rho} =  \nu_k+a_\rho$, and $\widetilde{b}_{k, \rho} = S_k+b_\rho$. 

\section{Dynamic Programming (DP) algorithm} \label{app:DP}
In this section, we describe the principle of the standard Dynamic Programming algorithm for a discrete segmentation problem and specify how to apply it on a finite and given grid as in our case. 

\paragraph{DP for discrete change-points problem} 
If we observe an ordered sequence of $A$ observations, the goal is to find a partition of the discrete grid $\llbracket 1,A \rrbracket$ into $K$ segments delimited by $K-1$ change-points denoted $a_k$ for $k=1,\ldots, K-1$ with the convention  $a_0=0$ and $a_K=A$. The $k$-th segment is $\llbracket a_{k-1}+1,a_k \rrbracket$ and we define the set of all possible segmentations with $K-1$ change-points:
$$
\mathcal{A}^K_A=\{\abf=(a_1,a_2,\ldots,a_{K-1}) \in {\mathbb{N}^{K-1}: a_0=0<a_1<\ldots <a_{K-1}<a_K=A}\},
$$
with cardinality $\binom{A-1}{K-1}$. Even though this space is finite, it is extremely large and a naive search is computationally prohibitive. The well-known solution consists of using the DP algorithm which can be applied if and only if the quantity to be optimized is segment-additive:
$$
\w{\abf}= (\w a_1,\w a_2,\ldots,\w a_{K-1})=\argmin{\abf \in \mathcal{A}^K_A} R_\abf= \argmin{\abf \in \mathcal{A}^K_A} \sum_{k=1}^K C(a_{k-1}+1:a_k),
$$
where $R_\abf$ is called the cost of the segmentation $\abf$ and $C(a_{k-1}+1:a_k)=C(\llbracket a_{k-1}+1, a_k\rrbracket)$ the cost of the segment $\llbracket a_{k-1}+1, a_k\rrbracket$. We define
$$
C_{K,A}=\minim{\abf \in \mathcal{A}^K_A}  \sum_{k=1}^K C(a_{k-1}+1:a_k),
$$
the cost of the best segmentation in $K$ segments. If we note $C(i:j)=C(\llbracket i, j\rrbracket)$ the cost of the segment $\llbracket i, j\rrbracket$, thanks to the segment-additivity property of $R_\abf$, DP solves the optimization problem using the following update rule: 
\begin{equation*}
  C_{K,A} = \minim{0<a_1<...<a_{K-1}<A} \sum_{k=1}^K C(a_{k-1}+1:a_k)= \minim{K-1 \leq h<A}  \left \{C_{K-1,h} +C(h+1:A) \right \}. 
\end{equation*}

This algorithm requires calculating the cost of each possible segment which is
\begin{align*}  
  C(i:j) 
  & = C(\llbracket i, j \rrbracket) \ \ \ \text{if $1 \leq i  \leq j \leq A$} \\
  & = + \infty \ \ \ \text{otherwise.} 
\end{align*}

\paragraph{Search in a fixed and finite grid} Let  
$$
\tbfp=\{\text{tp}_1, \text{tp}_2 ,   \ldots ,\text{tp}_A\},
$$
be this grid where $0<\text{tp}_i<1$ and $A$ is the size of the grid or the number of potential change-points. If a segment is defined as $(\text{tp}_\ell,\text{tp}_h]$ with $\ell <h$, we can apply DP with  
\begin{align*}  
  C(i:j) 
  & = C((\text{tp}_{i-1},\text{tp}_j]) \ \ \ \text{if $1 \leq i  \leq j \leq A+1$} \\
  & = + \infty \ \ \ \text{otherwise.} 
\end{align*}
with the convention $\text{tp}_0=0$ and $\text{tp}_{A+1}=1$. The complexity is thus $\mathcal{O}((A+1)^2 K)$ and the optimal change-points are given by 
$$
\w{\tau}_k=\text{tp}_{\w{a}_k}.
$$

\section{Cross-validation procedure} \label{app:cv}

\begin{algorithm}[H]
  \caption{Cross-validation procedure}
  \label{algo:cv}
  \begin{algorithmic}[1]
    \State {\bf input:} a realization of the process $N$.
    \State {\bf cross-validation:} for $m = 1$ to $M$
    \begin{itemize}
      \State {\bf sample} a learning process $N^{m, L}$ from $N$ with probability $f$, and form the test process $N^{m, T}$ with the remaining events;
      \State {\bf for $K = 1$ to $K_{\max}$, } 
      \begin{itemize}
        \State {\bf segment} the learning process $N^{m, L}$ using the Poisson-Gamma contrast \eqref{contrast:PG} to get
        $$
        \widehat{\taubf}^{m, K, L}
        =
        \argmin{\nubf \in  \Upsilon^{K,n(m,L)}_\star}   \min_{\taubf \in \cM_\nubf^K(N^{m,L})}  \gamma_{PG} (\taubf),
        $$
        where $n(m,L)$ is the number of events in the learning process,
        \State {\bf deduce} the estimate $\widehat{\lambdabf}^{m, K, L}$ as the set of posterior means \eqref{PosteriorMeanLambda}
        $$
        \widehat{\lambda}_k^{m, K, L} = \E(\lambda_k | N^{m, L}, \widehat{\taubf}^{m, K, L})= \frac{a^{m,L}+\Delta N_k^{m,L}}{b^{m,L}+\Delta \hat{\tau}_k^{m,K,L}} \ \ \ \ \text{for all $k = 1, \dots K$},
        $$ 
        \State {\bf compute} the Poisson contrast for the test process
        $$
        \gamma^{m, K, T} = - \log p_P\left(N^{m, K, T}; \widehat{\taubf}^{m, K, L}, \frac{1-f}f \widehat{\lambdabf}^{m, K, L}\right).
        $$
      \end{itemize}
    \end{itemize}
    \State {\bf averaging:} for $K = 1$ to $K_{\max}$, compute the average contrast
      $$
      \overline{\gamma}^{K, T} = \frac1M \sum_{m=1}^M \gamma^{m, K, T}.
      $$
    \State {\bf selection:} select $K$ as
      $$
      \widehat{K} = \argmin{K} \; \overline{\gamma}^{K, T}.
      $$
    \end{algorithmic}
\end{algorithm}

\section{Simulation parameters} \label{app:simulation}

\begin{figure}[hbtp]
  \centering
  \includegraphics[width=.8\textwidth, trim=0 0 0 10, clip=]{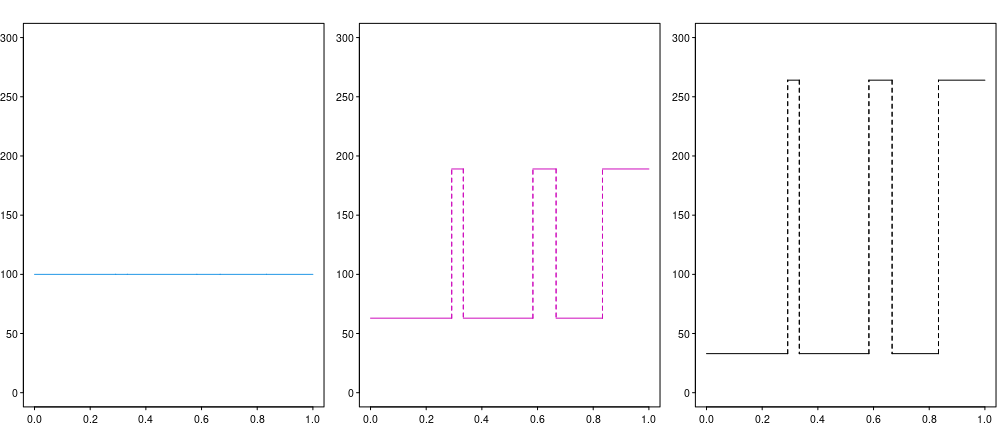}
  \caption{Piecewise intensity function $\lambda(t)$ on $[0,1]$ for $\bar{\lambda}=100$ and $\lambda_R\in \{1,3,8\}$ from left to right.}
  \label{fig:lambda}
\end{figure}

We provide here measures of the difficulty of the segmentation problems considered in the simulation study presented in Section \ref{sec:simu}. To this aim, we define a signal-to-noise ratio (SNR) measure as

\begin{equation} \label{eq:SNR}
\text{SNR}=\sqrt{\frac{\Delta\tau_- \lambda_-^2 + \Delta\tau_+ \lambda_+^2-\lambdabar^2}{\lambdabar}}
\end{equation} 
This definition of the SNR is consistent with this used in the Gaussian setting to measure the difficulty of a segmentation problem \citep[see e.g.][]{PLH11}. The SNR is zero when the intensity $\lambda$ is constant. Table \ref{table:lp_lm} gives the values of $\lambda_-$ and $\lambda_+$ for each combination $(\lambdabar, \lambda_R)$ considered in the simulation design described in Section \ref{sec:simudesign}. Table \ref{table:SNR} gives the corresponding SNR values. 

\begin{table}[hbtp]
  \caption{Values of $\lambda_-$, $\lambda_+$ for each combination $\lambdabar / \lambda_R$ used in the simulation study. The case $\lambda_R = 1$ ($\lambda_- = \lambda_+ = \lambdabar$) is not displayed.}\label{table:lp_lm}
  \begin{center}
    \begin{tabular}{cccccccc}
      \multicolumn{1}{c}{} & \multicolumn{7}{c}{$\lambda_R$} \\
      $\lambdabar$ &  2 & 3 & 4 &6 & 8 & 11 & 16 \\
      \hline
      32 & 25--50 & 20--60 & 17--68 & 13--78 & 11--88 & ~8--88 & ~6--96 \\
      56 & 43--86 & ~35--105 & ~30--120 & ~23--138 & ~18--144 & ~14--154 & ~10--160 \\
      100 & ~77--154 & ~63--189 & ~53--212 & ~41--246 & ~33--264 & ~26--286 & ~19--304 \\
      178 &  138--276 & 112--336 & ~95--380 & ~72--432 & ~59--472 & ~45--495 & ~33--528 \\
      316 &  245--490 & 200--600 & 169--676 & 129--774 & 104--832 & ~81--891 & ~59--944 \\
      562 &   435--870 & ~355--1065 & ~300--1200 & ~229--1374 & ~185--1480 & ~143--1573 & ~105--1680 \\
      1000 & ~774--1548 & ~632--1896 & ~533--2132 & ~407--2442 & ~329--2632 & ~255--2805 & ~186--2976 \\
    \end{tabular}
  \end{center}
\end{table}

\begin{table}[hbtp]
  \caption{Values of $SNR$ for the different values of $\lambdabar$ and $\lambda_R$ used in the simulation study.}\label{table:SNR}
  \begin{center}
    \begin{tabular}{cccccccc}
      \multicolumn{1}{c}{} & \multicolumn{7}{c}{$\lambda_R$} \\
      $\lambdabar$ &  2 & 3 & 4 &6 & 8 & 11 & 16 \\
      \hline
      32 & 2.15 & 3.11 & 4.07 & 5.21 & 6.42 & 6.32 & 7.27 \\ 
      56 & 2.43 & 4.11 & 5.51 & 7.06 & 7.49 & 8.37 & 8.87 \\ 
      100 & 3.34 & 5.68 & 7.14 & 9.40 & 10.54 & 11.97 & 13.12 \\ 
      178 & 4.75 & 7.54 & 9.72 & 12.18 & 14.17 & 15.22 & 16.83 \\ 
      316 & 6.34 & 10.29 & 13.03 & 16.56 & 18.63 & 20.77 & 22.68 \\ 
      562 & 8.33 & 13.62 & 17.28 & 22.00 & 24.86 & 27.35 & 30.28 \\ 
      1000 & 11.10 & 18.20 & 22.96 & 29.27 & 33.12 & 36.62 & 40.10 \\ 
    \end{tabular}
  \end{center}
\end{table}

To make the SNR measure even more interpretable, we relate it to the power of the Kolmogorov-Smirnoff (KS) for the uniform distribution of the event times over $(0, 1)$. In absence of signal (constant intensity $\lambda$), the events are uniformly distributed over $(0, 1)$. When the segmentation structure gets stronger (higher SNR), the KS is more and more capable of detecting a departure from the uniform distribution. \\
Figure \ref{fig:SNR1_PowerUnif} shows how the power of the KS increases as the SNR increases. We observe that for SNR smaller than 10 (see Table \ref{table:SNR}), the power of the KS test is low, meaning simply detecting heterogeneity is a difficult task. Note that detecting heterogeneity is a much simpler task than precisely locating change-points and estimating the intensity between each of them. We see that the configurations considered in the proposed simulation design range from very difficult (very low power of the KS test) to fairly easy.

\begin{figure}[hbtp]
  \centering
  \includegraphics[width=.5\textwidth, trim=0 0 20 30, clip=]{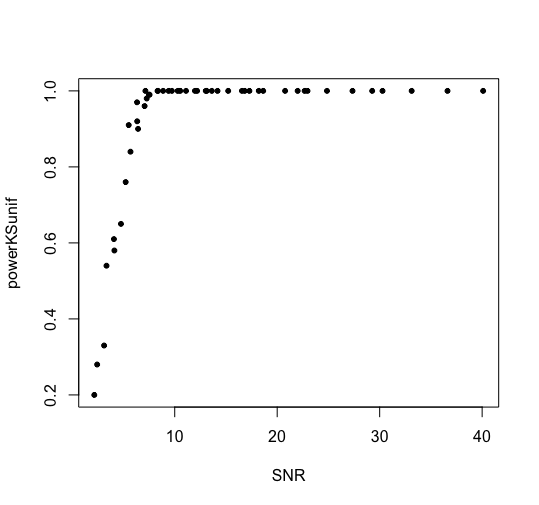}
  \caption{Power of the Kolmogorov-Smirnoff (KS) for the uniform distribution of the event times, as a function of the SNR. Each dot corresponds to a given combination $(\lambdabar, \lambda_R)$ used in the simulation.}
  \label{fig:SNR1_PowerUnif}
\end{figure}

\section{Simulation of PHP}\label{app:simuHP}
We present here an algorithm for sampling from the presented PHP.
The algorithm we employ is the modified thinning algorithm proposed by \cite{Ogata1981} for sampling from an HP process \citep[see also][]{DaleyAl2006,Rasmussen2018,LaubAL2015}.

Consider a PHP with conditional intensity function $\lambda(t)$. Conceptually, the thinning algorithm we use to sample from $N$ in $[0, T]$ resembles an accept/reject mechanism: initially generating more points than required in $[0, T]$ using a larger intensity than $\lambda$. In other words, points are sampled according to a Poisson process of intensity function $M(t)$ satisfying $M(t) > \lambda(t)$ for all $t \in [0,T]$. 
Subsequently, some of these points are discarded with a certain probability, ensuring that the retained points follow the desired distribution. 
Precisely, each point is deleted with a probability $1 - \lambda/M$ 
The retained points form a sample from a point process of conditional intensity $\lambda(t)$; see \citep[Proposition 1]{Ogata1981}.

\begin{algorithm}[H]
    \caption{Sampling from a point process of conditional intensity $\lambda$ and history $H_t$}
    \label{algo:sampling_by_thinning}
    \begin{algorithmic}[1]
      \State{\textbf{input} conditional intensity $\lambda$}
      \State \textbf{set} t=0, i=0 
      \While{$t \leq T$}
        \State {\textbf{compute} $M(t)$ via Equation \eqref{eq:bound_c_in_for_simulation}}
        \State {\textbf{generate} independent r.v. $s \sim \text{Exp}(1/M(t))$ and $u \sim \text{Unif}(0,1)$}
        \State {\textbf{set} $p= \lambda(t + s \mid H_{t+s})/M(t)$ the retaining probability of the proposed point $t+s$}
        \If{$t+s<T$ and $u \leq p$} 
            \State{\textbf{update} $i=i+1$}
            \State{\textbf{set} $t_i=t+s$ (accept the proposed point)}
        \EndIf
        \State{\textbf{update} $t=t+s$}
      \EndWhile
     \State {\textbf{return } $\{t_i\}_{i}$}
    \end{algorithmic}
 \end{algorithm}

Note that the selection of the function $M$ plays a crucial role in the simulation process. 
On one hand, $M$ should satisfy $M(t) > \lambda(t )$ for all $t \in [0, T]$ ensuring that the intensity of proposed points exceeds the conditional intensity of the PHP.
On the other hand, $M$ needs to be close to $\lambda$ to maximize the simulation's efficiency, thereby minimizing the number of rejected points.
Following~\cite{Ogata1981}, we opt for the function $M$ defined for $t \in [0, T]$ by
\begin{equation}
    \label{eq:bound_c_in_for_simulation}
    M(t) := \lambda(t) + \alpha \sum_{k=1}^K c_k \mathds{1}_{[\tau_{k-1}, \tau_k)}(t).
\end{equation}

\section{Additional simulations} \label{app:appSimuls}
This section is dedicated to the study of the robustness of the proposed method to the two tuning parameters that are, 
\begin{itemize}
  \item[-] the hyper-parameters $(a, b)$ used in the Poisson-Gamma contrast defined in Equation \eqref{contrast:PG} and
  \item[-] the sampling fraction $f$ for the cross-validation procedure described in Section \ref{sec:selection}.
\end{itemize}
For both parameters, we use the simulation design described in Section \ref{sec:simudesign}, with the two typical parameters $\lambdabar = 100$ and $\lambda_R = 8$. These values are chosen, as they yield in average performances for the proposed algorithm, as shown in Figure \ref{fig:simPP-HausL2cum}.

\paragraph{Hyper-parameters of the Poisson-Gamma contrast} 
Regarding the choice of the hyper-parameters $a$ and $b$, we stick to the rule $a/b = n$, which guarantees that the expectation of $N(T)$ is $n$. Hence, we are left with the choice of $a = nb$, reminding us that the variance of the Gamma distribution $\Gam(a, b)$ is $a/b^2$, which increases as $a$ decreases (as $b=a/n$). We consider five values for $a$: $a = 0.1, 0.5, 1, 2$ and $10$.

\begin{figure}[ht]
  \begin{center}
    \begin{tabular}{ccc}
      $\widehat{K}$ & $d(\widehat{\taubf}, \taubf)$ & $\ell_2(\widehat{\Lambda}, \Lambda)$ \\
      \includegraphics[width=.3\textwidth, trim=10 10 10 10 50, clip=]{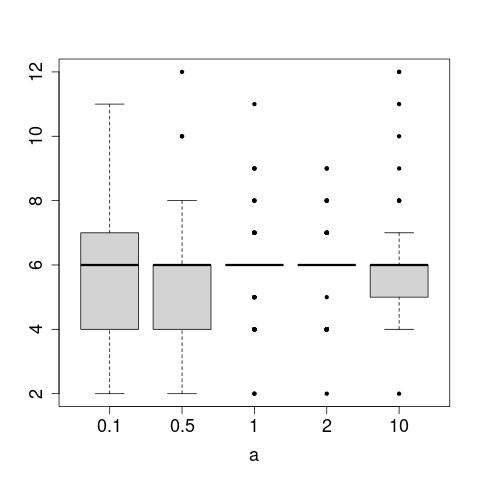} & 
      \includegraphics[width=.3\textwidth, trim=10 10 10 10 50, clip=]{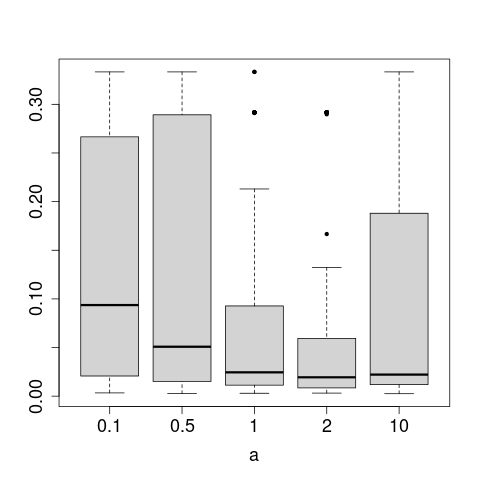} & 
      \includegraphics[width=.3\textwidth, trim=10 10 10 10 50, clip=]{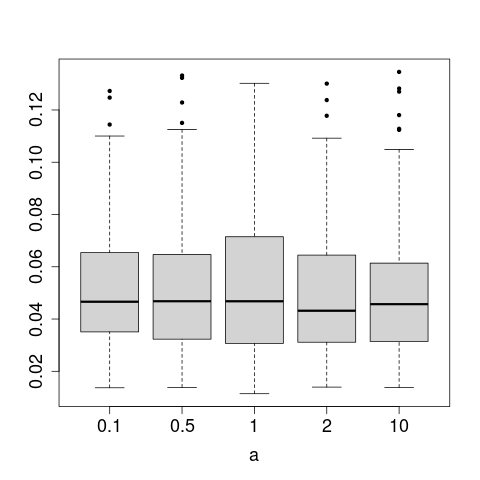}     
    \end{tabular}
    \caption{Effect of the choice of the hyper-parameter $a$ on three evaluation criteria for $\lambdabar = 100$, $\lambda_R = 8$ and for a true number of segments $K=6$.
    Left: selected number of segments $\widehat{K}$ using cross-validation; 
    Center: Hausdorff distance between the estimated change-points $\widehat{\taubf}$ and the true ones $\taubf$; 
    Right: relative $L^2$-norm between the estimated and the true cumulated intensity functions.
    \label{fig:simPP-aRobust}}
  \end{center}
\end{figure}

Figure \ref{fig:simPP-aRobust} (left) shows that the averaged selected $\widehat{K}$ is close to the true one ($K=6$) for all values of $a$ and that the suggested value $a = 1$ yields the lower variance. The center panel of the same figure shows that the best precision for the change-point locations, measured by the Hausdorff distance $d(\widehat{\taubf}, \taubf)$, is obtained for $a=1$ or $2$. Finally, the right panel shows that the relative $L^2$ distance between the estimated and true cumulated intensity is fairly insensitive to the choice of $a$. Overall, this suggests that the recommendation $(a=1, b=1/n)$ is a reasonable and simple choice for the hyper-parameters.

\paragraph{Sampling fraction for cross-validation} 
In the simulation study presented in Section \ref{sec:simu}, a fraction $f = 4/5$ of the data is used to build each learning set. We consider here the fractions $f = 1/2, 2/3, 4/5$ and $9/10$.

\begin{figure}[ht]
  \begin{center}
    \begin{tabular}{ccc}
      $\widehat{K}$ & $d(\widehat{\taubf}, \taubf)$ & $\ell_2(\widehat{\Lambda}, \Lambda)$ \\
      \includegraphics[width=.3\textwidth, trim=10 10 10 10 50, clip=]{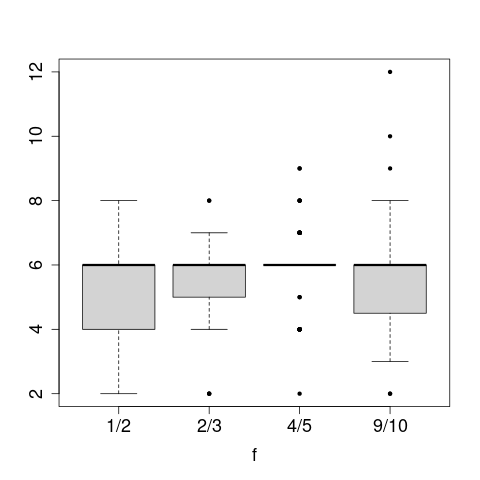} & 
      \includegraphics[width=.3\textwidth, trim=10 10 10 10 50, clip=]{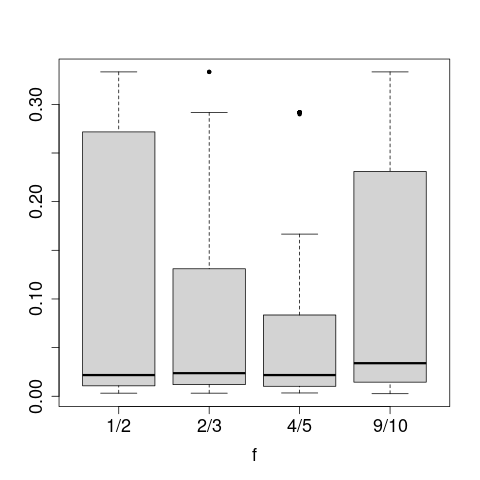} & 
      \includegraphics[width=.3\textwidth, trim=10 10 10 10 50, clip=]{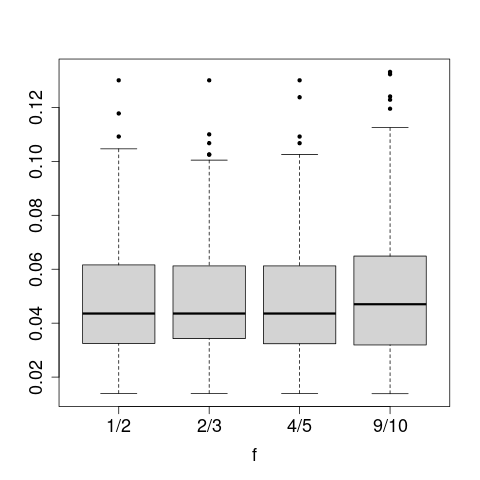}     
    \end{tabular}
    \caption{Effect of the choice of the fraction $f$ used in the cross-validation procedure described in Section \ref{sec:selection} on three evaluation criteria for $\lambdabar = 100$, $\lambda_R = 8$ and for a true number of segments $K=6$. Left, center and right: same legend as Figure \ref{fig:simPP-aRobust}.
    \label{fig:simPP-vRobust}}
  \end{center}
\end{figure}

Figure \ref{fig:simPP-vRobust} (left) shows that all values of $f$ give the right number of segments $K$ in average, but that the recommended value $f = 4/5$ yields the smallest variance. It also yields the smallest Hausdorff distance between the estimated and true change-points (center), while the relative $L^2$ distance between the estimate $\widehat{\Lambda}$ and the true $\Lambda$ is unaffected by the choice of $f$. The fraction $f = 4/5$ therefore seems to be a reasonable choice for the cross-validation in our case.

\section{Real life data sets illustrations} \label{app:appliPP}

\begin{figure}[hbtp]
  \centering
  \includegraphics[width=.35\textwidth, trim=0 0 20 30, clip=]{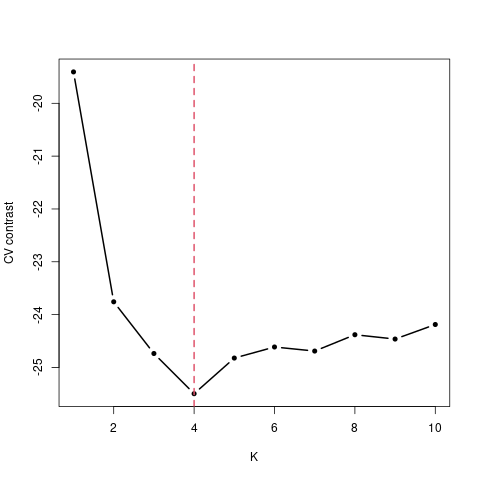} 
  \qquad
  \includegraphics[width=.35\textwidth, trim=0 0 20 30, clip=]{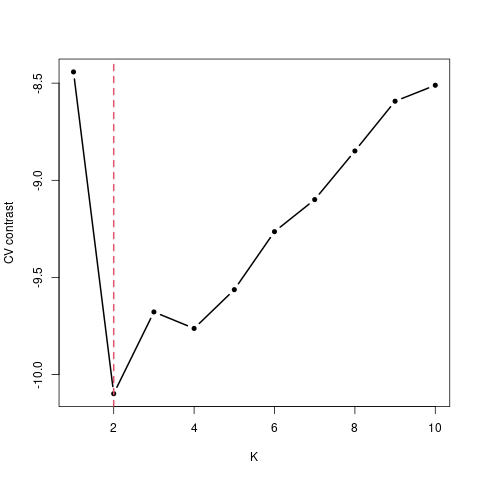}
  \caption{Eruptions of the Kilauea (left) et Mauna Loa (right) volcanos. Selection of the number of segments $K$: black solid-bullet line = criterion $\overline{\gamma}^{K, T}$, vertical red dotted line = optimal number of segments $\widehat{K}$. \label{fig:appPPselection}}
\end{figure}

\begin{figure}[hbtp]
  \centering
  \includegraphics[width=.35\textwidth, trim=0 0 20 30, clip=]{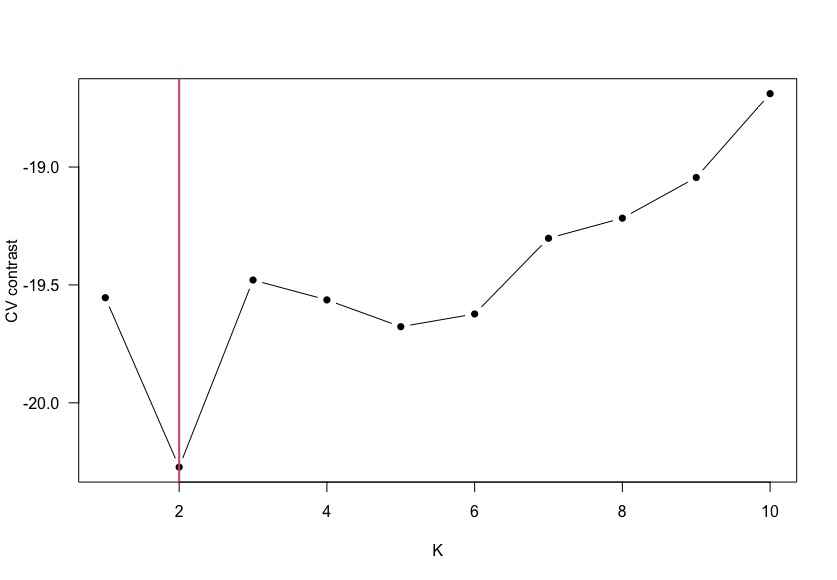}
  \qquad
  \includegraphics[width=.35\textwidth, trim=0 0 20 30, clip=]{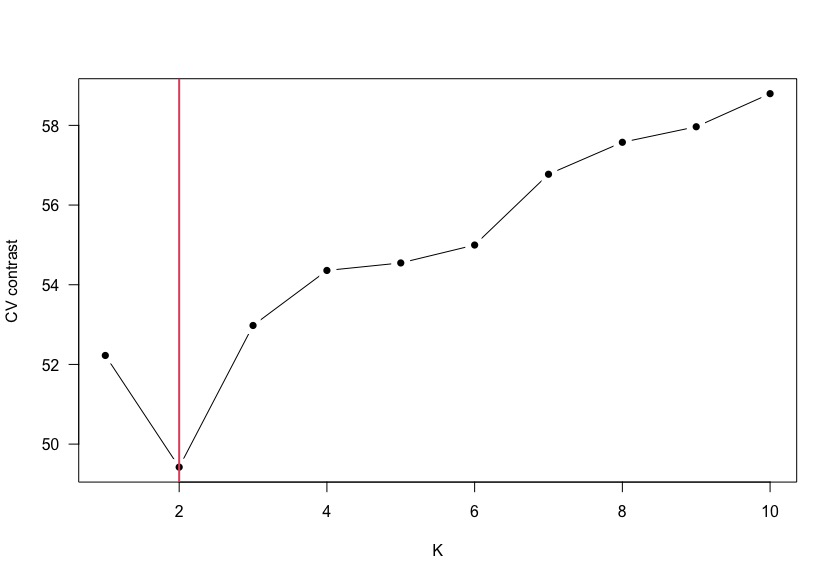}
  \caption{
  Eruptions of Etna.  Left: Poisson model. Right: marked Poisson model.
  Contrast function as a function of $K$, and in red the selected dimension $\w{K}=2$ in both cases.
  }
\label{fig:etna-contrasts}
\end{figure}

Figure \ref{fig:appPPappendix} displays the segmentations of the Mauna Loa series obtained by considering $K = 3, 4, 5$ and $6$ segments. The fit obviously improves as $K$ increases and the cross-validation procedure proposed in Section \ref{sec:selection} aims to avoid overfitting. Exploring higher values of $K$ may still be useful for uncovering tiny heterogeneity, which might merit further investigation (see, e.g., the 1870-1880 segment exhibited with $K=6$).

\begin{figure}[hbtp]
\centering
  $\begin{array}{cc}
    K=3 & K=4 \\
    \includegraphics[width=.35\textwidth, trim=0 0 20 30, clip=]{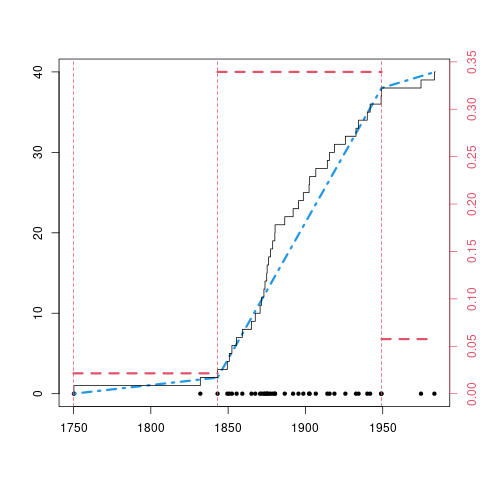} &
    \includegraphics[width=.35\textwidth, trim=0 0 20 30, clip=]{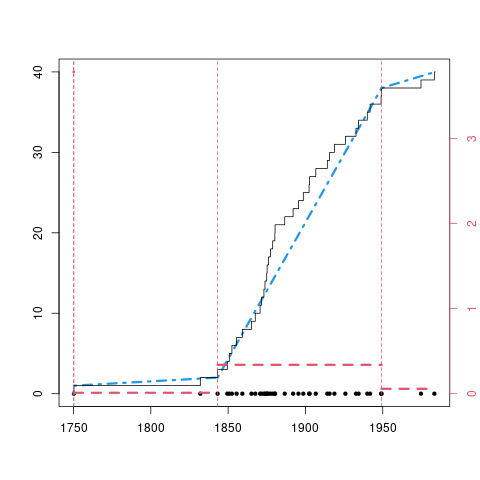} \\
        K=5 & K=6 \\
    \includegraphics[width=.35\textwidth, trim=0 0 20 30, clip=]{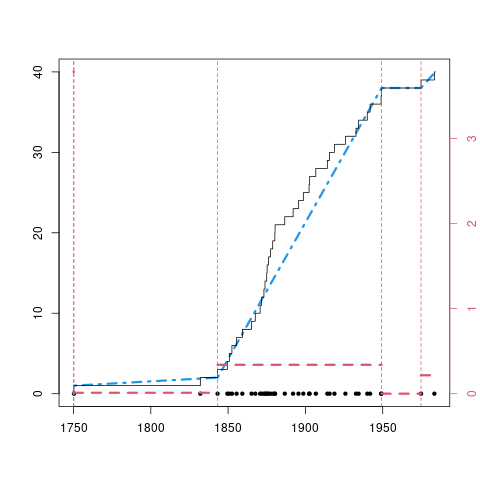} &
    \includegraphics[width=.35\textwidth, trim=0 0 20 30, clip=]{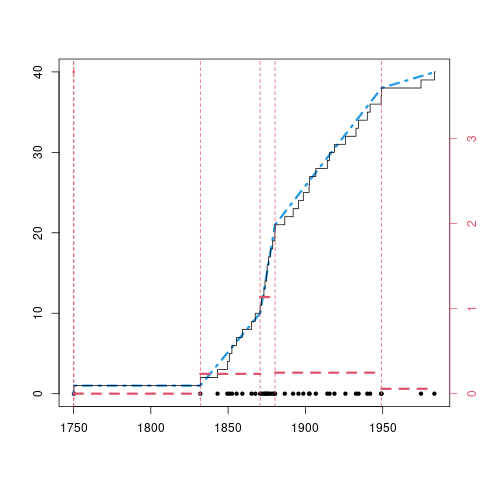}
  \end{array}$
  \caption{Eruptions of the Mauna Loa volcano. Segmentation with $K = 3, 4, 5, 6$ segments. Same legend as Figure \ref{fig:appPPsegmentation}. Note that the scale of the right $y$-axis differs between plots. The first segment for $K = 4, 5$ and $6$ is included in the first year of observation (1750) and contains only one event. \label{fig:appPPappendix}}
\end{figure}


\end{document}